\documentclass[11pt, fullpage]{article}	
\usepackage[margin=1in]{geometry}
\usepackage{url}
\usepackage{float} 

\usepackage{amssymb,amsmath,amsfonts,graphicx,psfig}
\usepackage{times}
\usepackage{xspace}
\usepackage{subfigure} 
\usepackage{xcolor}
\usepackage[ruled]{algorithm2e}

\usepackage{color}
\usepackage{xcolor}
\usepackage{amsthm}

{\makeatletter
 \gdef\xxxmark{%
   \expandafter\ifx\csname @captype\endcsname\relax % based on a TeXBook e.g.
     \marginpar{xxx}% not in a caption, can use marginpar
   \else
     xxx % notice trailing space
   \fi}
 \gdef\xxx{\@ifnextchar[\xxx@lab\xxx@nolab}
 \long\gdef\xxx@lab[#1]#2{{\bf [\xxxmark #2 ---{\sc #1}]}}
 \long\gdef\xxx@nolab#1{{\bf [\xxxmark #1]}}
}

\newcommand{\CASE}[1]{\STATE \textbf{case} #1\textbf{:} \begin{ALC@g}}
\newcommand{\ENDCASE}{\end{ALC@g}}

\newcommand{\DEFAULT}{\STATE \textbf{default:} \begin{ALC@g}}
\newcommand{\ENDDEFAULT}{\end{ALC@g}}
\newcommand{\DEFAULTLINE}[1]{\STATE \textbf{default:} }

\newtheorem{lemma}{Lemma}
\newtheorem{theorem}{Theorem}   
\newtheorem{corollary}{Corollary}

\newcommand*\Let[2]{#1 $\gets$ #2}

\newcommand{\opt}{\ensuremath{\operatorname{\textsc{Opt}}}\xspace}
\newcommand{\Opt}{\ensuremath{\operatorname{\textsc{Opt}}}\xspace}
\newcommand{\OPT}{\ensuremath{\operatorname{\textsc{Opt}}}\xspace}
\newcommand{\BF}{\ensuremath{\operatorname{\textsc{Bf}}}\xspace}
\newcommand{\IFF}{\ensuremath{\operatorname{\textsc{Iff}}}\xspace}
\newcommand{\RFF}{\ensuremath{\operatorname{\textsc{Rff}}}\xspace}
\newcommand{\RH}{\ensuremath{\operatorname{\textsc{Rh}}}\xspace}

\newcommand{\FF}{\ensuremath{\operatorname{\textsc{Ff}}}\xspace}
\newcommand{\NF}{\ensuremath{\operatorname{\textsc{Nf}}}\xspace}
\newcommand{\HA}{\ensuremath{\operatorname{\textsc{Ha}}}\xspace}
\newcommand{\MH}{\ensuremath{\operatorname{\textsc{Mh}}}\xspace}
\newcommand{\HM}{\ensuremath{\operatorname{\textsc{Hm}}}\xspace}
\newcommand{\HMM}{\ensuremath{\operatorname{\textsc{Hmm}}}\xspace}
\newcommand{\RHM}{\ensuremath{\operatorname{\textsc{Rhm}}}\xspace}
\newcommand{\OM}{\ensuremath{\operatorname{\textsc{Om}}}\xspace}
\newcommand{\ROM}{\ensuremath{\operatorname{\textsc{Rom}}}\xspace}
\newcommand{\RRM}{\ensuremath{\operatorname{\textsc{Rrm}}}\xspace}
\newcommand{\MFF}{\ensuremath{\operatorname{\textsc{Mff}}}\xspace}
\newcommand{\MBF}{\ensuremath{\operatorname{\textsc{Mbf}}}\xspace}

\newcommand{\algo}[1]{{#1}}{}

\newcommand{\longv}[1]{\hspace{-.11cm}}

\newcommand{\red}[1]{}
\newcommand{\redd}[1]{}
\newcommand{\reds}[1]{}
\newcommand{\XXXX}[1]{}
\newcommand{\cmnt}[1]{}

\usepackage{hyperref}

\begin{document}

\title{An All-Around Near-Optimal Solution for the \\ Classic Bin Packing Problem}

\author{
Shahin Kamali \and Alejandro L\'opez-Ortiz \thanks{
School of Computer Science,
University of Waterloo, Canada.
Email: {\tt \{s3kamali,alopez-o\}@uwaterloo.ca} }
}
\date{}
\pagestyle{plain}

\maketitle

%%%%%%%%%%%%%%%%%%%%%%%%%%%%%%%%%%%%%%%%%%%%%%%%%%%%%%%%%%%%%%%%%%%%%%%%%%%%%%%%%%%%%%%%%%%%%%%%%%%%%%%%%%%%%%%%%%%%%%%%%%%%%%%%%%%%%%%%%%%%%%%%%%%%%%%%%%%%%%%%%%%%%
%%%%%%%%%%%%%%%%%%%%%%%%%%%%%%%%%%%%%%%%%%%%%%%%%%%%%%%%%%%%%%%%%%%%%%%%%%%%%%%%%%%%%%%%%%%%%%%%%%%%%%%%%%%%%%%%%%%%%%%%%%%%%%%%%%%%%%%%%%%%%%%%%%%%%%%%%%%%%%%%%%%%%
%%%%%%%%%%%%%%%%%%%%%%%%%%%%%%%%%%%%%%%%%%%%%%%%%%%%%%%%%%%%%%%%%%%%%%%%%%%%%%%%%%%%%%%%%%%%%%%%%%%%%%%%%%%%%%%%%%%%%%%%%%%%%%%%%%%%%%%%%%%%%%%%%%%%%%%%%%%%%%%%%%%%%
%%%%%%%%%%%%%%%%%%%%%%%%%%%%%%%%%%%%%%%%%%%%%%%%%%%%%%%%%%%%%%%%%%%%%%%%%%%%%%%%%%%%%%%%%%%%%%%%%%%%%%%%%%%%%%%%%%%%%%%%%%%%%%%%%%%%%%%%%%%%%%%%%%%%%%%%%%%%%%%%%%%%%
%%%%%%%%%%%%%%%%%%%%%%%%%%%%%%%%%%%%%%%%%%%%%%%%%%%%%%%%%%%%%%%%%%%%%%%%%%%%%%%%%%%%%%%%%%%%%%%%%%%%%%%%%%%%%%%%%%%%%%%%%%%%%%%%%%%%%%%%%%%%%%%%%%%%%%%%%%%%%%%%%%%%%
%%%%%%%%%%%%%%%%%%%%%%%%%%%%%%%%%%%%%%%%%%%%%%%%%%%%%%%%%%%%%%%%%%%%%%%%%%%%%%%%%%%%%%%%%%%%%%%%%%%%%%%%%%%%%%%%%%%%%%%%%%%%%%%%%%%%%%%%%%%%%%%%%%%%%%%%%%%%%%%%%%%%%
%%%%%%%%%%%%%%%%%%%%%%%%%%%%%%%%%%%%%%%%%%%%%%%%%%%%%%%%%%%%%%%%%%%%%%%%%%%%%%%%%%%%%%%%%%%%%%%%%%%%%%%%%%%%%%%%%%%%%%%%%%%%%%%%%%%%%%%%%%%%%%%%%%%%%%%%%%%%%%%%%%%%%
%%%%%%%%%%%%%%%%%%%%%%%%%%%%%%%%%%%%%%%%%%%%%%%%%%%%%%%%%%%%%%%%%%%%%%%%%%%%%%%%%%%%%%%%%%%%%%%%%%%%%%%%%%%%%%%%%%%%%%%%%%%%%%%%%%%%%%%%%%%%%%%%%%%%%%%%%%%%%%%%%%%%%

\begin{abstract}
In this paper we present the first algorithm with optimal average-case and close-to-best known worst-case performance for the classic on-line problem of bin packing. It has long been observed that known bin packing algorithms with optimal average-case performance were not optimal in the worst-case sense. In particular First Fit and Best Fit had optimal average-case ratio of 1 but a worst-case competitive ratio of 1.7. 
The wasted space of First Fit and Best Fit for a uniform random sequence of length $n$ is expected to be $\Theta(n^{2/3})$ and $\Theta(\sqrt{n} \log ^{3/4} n)$, respectively. The competitive ratio can be improved to 1.691 using the Harmonic algorithm; further variations of this algorithm can push down the competitive ratio to 1.588. However, Harmonic and its variations have poor performance on average; in particular, Harmonic has average-case ratio of around 1.27. In this paper, first we introduce a simple algorithm which we term Harmonic Match. This algorithm performs as well as Best Fit on average, i.e., it has an average-case ratio of 1 and expected wasted space of $\Theta(\sqrt{n} \log ^{3/4} n)$. Moreover, the competitive ratio of the algorithm is as good as Harmonic, i.e., it converges to $ 1.691$ which is an improvement over 1.7 of Best Fit and First Fit. 
We also introduce a different algorithm, termed as Refined Harmonic Match, which achieves an improved competitive ratio of $1.636$ while maintaining the good average-case performance of Harmonic Match and Best Fit.
Finally, our extensive experimental evaluation of the studied bin packing algorithms shows that our proposed algorithms have comparable average-case performance with Best Fit and First Fit, and this holds also for sequences that follow distributions other than the uniform distribution.

\end{abstract}

%%%%%%%%%%%%%%%%%%%%%%%%%%%%%%%%%%%%%%%%%%%%%%%%%%%%%%%%%%%%%%%%%%%%%%%%%%%%%%%%%%%%%%%%%%%%%%%%%%%%%%%%%%%%%%%%%%%%%%%%%%%%%%%%%%%%%%%%%%%%%%%%%%%%%%%%%%%%%%%%%%%%%
%%%%%%%%%%%%%%%%%%%%%%%%%%%%%%%%%%%%%%%%%%%%%%%%%%%%%%%%%%%%%%%%%%%%%%%%%%%%%%%%%%%%%%%%%%%%%%%%%%%%%%%%%%%%%%%%%%%%%%%%%%%%%%%%%%%%%%%%%%%%%%%%%%%%%%%%%%%%%%%%%%%%%
%%%%%%%%%%%%%%%%%%%%%%%%%%%%%%%%%%%%%%%%%%%%%%%%%%%%%%%%%%%%%%%%%%%%%%%%%%%%%%%%%%%%%%%%%%%%%%%%%%%%%%%%%%%%%%%%%%%%%%%%%%%%%%%%%%%%%%%%%%%%%%%%%%%%%%%%%%%%%%%%%%%%%
%%%%%%%%%%%%%%%%%%%%%%%%%%%%%%%%%%%%%%%%%%%%%%%%%%%%%%%%%%%%%%%%%%%%%%%%%%%%%%%%%%%%%%%%%%%%%%%%%%%%%%%%%%%%%%%%%%%%%%%%%%%%%%%%%%%%%%%%%%%%%%%%%%%%%%%%%%%%%%%%%%%%%
%%%%%%%%%%%%%%%%%%%%%%%%%%%%%%%%%%%%%%%%%%%%%%%%%%%%%%%%%%%%%%%%%%%%%%%%%%%%%%%%%%%%%%%%%%%%%%%%%%%%%%%%%%%%%%%%%%%%%%%%%%%%%%%%%%%%%%%%%%%%%%%%%%%%%%%%%%%%%%%%%%%%%
%%%%%%%%%%%%%%%%%%%%%%%%%%%%%%%%%%%%%%%%%%%%%%%%%%%%%%%%%%%%%%%%%%%%%%%%%%%%%%%%%%%%%%%%%%%%%%%%%%%%%%%%%%%%%%%%%%%%%%%%%%%%%%%%%%%%%%%%%%%%%%%%%%%%%%%%%%%%%%%%%%%%%
%%%%%%%%%%%%%%%%%%%%%%%%%%%%%%%%%%%%%%%%%%%%%%%%%%%%%%%%%%%%%%%%%%%%%%%%%%%%%%%%%%%%%%%%%%%%%%%%%%%%%%%%%%%%%%%%%%%%%%%%%%%%%%%%%%%%%%%%%%%%%%%%%%%%%%%%%%%%%%%%%%%%%
%%%%%%%%%%%%%%%%%%%%%%%%%%%%%%%%%%%%%%%%%%%%%%%%%%%%%%%%%%%%%%%%%%%%%%%%%%%%%%%%%%%%%%%%%%%%%%%%%%%%%%%%%%%%%%%%%%%%%%%%%%%%%%%%%%%%%%%%%%%%%%%%%%%%%%%%%%%%%%%%%%%%%
\thispagestyle{empty} 
\newpage

\pagenumbering{arabic} 

\section{Introduction}
An instance of the classical online bin packing problem is defined by a sequence of \textit{items} which are revealed in an online manner. Each item has a size in the range $(0,1]$. The goal is to pack these items into a minimum number of bins which have a uniform capacity of 1. A natural algorithm for the problem is \algo{Next Fit (\NF)} which keeps one \textit{open} bin at each time. If a given item does not fit into the bin, the algorithm \textit{closes} the bin (i.e., it does not refer to it in future) and opens a new bin. In contrast to \NF, \algo{First Fit (\FF)} does not close any bin: It maintains the bins in the order they are opened and places a given item in the first bin which has enough space for it. In case such a bin does not exist, it opens a new bin for the item. \algo{Best Fit (\BF)} performs similarly to \FF, except that it maintains the bins in decreasing order of their \emph{levels}; the level of a bin is the total size of items placed in the bin. %increasing order of their remaining capacities. 
Another approach is to divide items into a constant number of classes based on their sizes and pack items of the same class apart from other classes. An example is the Harmonic (\HA) algorithm which has a parameter $K$ and defines $K$ intervals $(1/2,1], (1/3,1/2], \ldots, (1/(K-1),1/K]$, and $(0,1/K]$; items which belong to the same interval are separately treated using the \algo{Next Fit} strategy. 

Online bin packing algorithms are usually compared through their respective average-case performance and worst-case performance. Under average-case analysis, it is assumed that item sizes follow a fixed distribution that is typically a uniform distribution. With this assumption, one can define the \textit{average performance ratio} as the ratio between the expected cost of an online algorithm for a randomly selected sequence compared to the cost of \OPT. Here \OPT is an optimal offline algorithm with unbounded computational power. It is known that \NF has an average performance ratio of $1.\bar{3}$ \cite{CoHoSY80} for sequences generated uniformly at random.
\FF and \BF are optimum in this sense and have an average ratio of 1 \cite{BeJLMM84}. To further compare algorithms with average ratio of 1, one can consider the \textit{expected waste} which is the expected amount of wasted space for serving a sequence of length $n$. More precisely, the wasted space of an algorithm for serving a sequence $\sigma$ is the difference between the cost of the algorithm and the total size of items in $\sigma$. \FF and \BF have expected waste of sizes $\Theta(n^{2/3})$ and $\Theta(\sqrt{n} \log ^{3/4} n)$, respectively \cite{Shor86,CoJoSW95,LeiSho89}. It is also known that all online algorithms have expected waste of size $\Omega(\sqrt{n} \lg ^{1/2} n)$ \cite{Shor86}. These results show that \BF is almost the best online algorithm with respect to average performance. %, at least when sequences are generated uniformly at random. 

There are other algorithms which perform almost as well as \BF on average. These algorithms are based on matching a `large' item with a `small' item to place them in the same bin. 
Throughout the paper we call an item \textit{large} if it is larger than 1/2 and \textit{small} otherwise.
Among the matching-based algorithms are \algo{Interval First Fit} (\IFF) \cite{CsiriGalam86} and \algo{Online Match} (\OM) \cite{CofLue91}. 
\IFF has a parameter $K$ and divides the unit interval into $K$ intervals of equal length, namely $I_t=(\frac{t-1}{K},\frac{t}{K}]$ for $t= 1,2, \ldots, K$. Here, $K$ is an odd integer and we have $K=2j+1$. The algorithms defines $j+1$ classes so that intervals $I_c$ and $I_{k-c}$ form class $c$ ($1 \leq c \leq j$) and interval $I_k$ forms class $j+1$. Items in each class are packed separately from other classes. The items in class $c$ $(2 \leq c \leq j+1)$ are treated using \FF strategy, while the items in the first class are treated using an Almost \FF strategy. Almost \FF is similar to \FF except that it closes a bin when it includes a small and a large item; 
%(when an item in interval $I_1$ is matched with an item in interval $I_{k-1}$); 
further, a large item is never placed in a bin which includes more than one small items, and a bin with $k$ small items is declared as being closed. The average ratio of \IFF is 1; precisely, it has an expected waste of $\Theta(n^{2/3})$.
%The algorithm matches items in interval $I_t$ with those in interval $I_{K-t}$. 
Algorithm \OM has also a parameter $K$ and declares two items as being \textit{companions} if their sum is in the range $[1-\frac{1}{K},1]$. To place a large item, \OM opens a new bin. To place a small item $x$, the algorithm checks whether there is an open bin $\beta$ with a large companion of $x$; in case there is, \OM places $x$ in $\beta$ and closes $\beta$. Otherwise, it packs $x$ using a \NF strategy in a separate list of bins. The average ratio of \OM converges to 1 for large values of $K$ \cite{CofLue91}. 

Although the matching-based algorithms have acceptable average performance, they do not perform well in the worst-case. In particular, \IFF has an unbounded competitive ratio \cite{CsiriGalam86}, and the competitive ratio of \OM is 2 \cite{CofLue91}. Among other matching algorithms we might mention \algo{Matching Best Fit (\MBF)} which performs similarly to \BF except that it closes a bin as soon as it receives the first small item. The average ratio of \MBF is as good as \BF while it has unbounded competitive ratio \cite{Shor86}. There is another algorithm which has expected waste of size $\Theta(\sqrt{n} \lg ^{1/2} n)$ \cite{Shor91} which %is even better than \BF and 
matches the lower bound of \cite{Shor86}. This algorithm also has a non-constant competitive ratio \cite{CoGaJo97}.

The competitive ratio\footnote{By competitive ratio, we mean \emph{asymptotic} competitive ratio where the number of opened bins by an optimal offline algorithm is arbitrarily large. For the results related to the \emph{absolute} competitive ratio, we refer the reader to \cite{CoGaJo97,BPsurvey2013}.} reflects the worst-case performance of online algorithms. More formally, it is the asymptotically maximum ratio between the cost of an online algorithm and that of \Opt for serving the same sequence. It is known that \NF has a competitive ratio of 2 while \FF and \BF have the same ratio of 1.7 \cite{JoDUGG74}. The competitive ratio of \HA converges to 1.691 for sufficiently large values of $K$ \cite{LeeLee85}. To be more precise, it approaches $T_\infty = \sum_1^\infty\frac{1}{t_i-1}$, where $t_1=2$ and $t_{i+1} = t_i(t_i-1)+1, i>1$ \footnote{Some notations are borrowed from \cite{CoGaJo97}.}. There are online algorithms which have even better competitive ratios. These include Modified First Fit (\MFF) with ratio 1.666 \cite{Yao80A}, Modified Harmonic with ratio around 1.635 \cite{LeeLee85}, and Harmonic++ with ratio $1.588$ \cite{Seid02}. These algorithms are members of a general framework of Super Harmonic algorithms \cite{Seid02}. Similar to \HA, Super Harmonic algorithms classify items by their sizes and pack items of the same class together. However, to handle the bad sequences of \HA, a fraction of opened bins include items from different classes. These bins are opened with items of small sizes in the hopes of subsequently adding items of larger sizes. At the time of opening such a bin, it is pre-determined how many items from each class should be placed in the bin. As the algorithms runs, the reserved spot for each class is occupied by an item of that class. It is guaranteed that the reserved spot is enough for any member of the class. This implies that the expected total size of items in the bin is strictly less than 1 by a positive value. Consequently, the expected waste of the algorithm is linear to the number of opened bins. Hence, for a sequence of length $n$, these algorithm have an expected waste of $\Omega(n)$. Since the expected wasted space of \opt is $o(n)$, the average performance ratio of Super Harmonic algorithms is strictly larger than 1. 
In particular, \algo{Refined Harmonic} and \algo{Modified Harmonic} have average performance ratios around 1.28 and 1.18, respectively \cite{GuChXu02,RamaTsuga89}.

Table \ref{competitiveTable} shows the existing results for major bin packing algorithms. 
As pointed out by Coffman et al. in \cite{CoGaJo97}, \textit{` All algorithms that do better than First Fit in the worst-case seem to do much worse in the average case.'} In this paper, however, we show that this is not a necessary condition and give an algorithm whose average-case ratio, competitive ratio, and expected wasted space are all at or near the top of each class. This also addresses a conjecture in \cite{GuChXu02} stated as \textit{`Harmonic-K is better than First Fit in the worst-case performance,
and First Fit is better than Harmonic-K in the average-case performance. Maybe there exists
an on-line algorithm with the advantages of both First Fit and Harmonic-K.'}

\subsection*{Contribution}

We introduce an algorithm called Harmonic Match (\HM) and show this algorithm is better than \BF and \FF in the worst case, while it performs as well as \BF and \FF on average. In particular, we show the competitive ratio of \HM is as good as \HA, i.e., it approaches $T_\infty \approx 1.691$ for sufficiently large values of $K$. For sequences generated uniformly at random, the average performance ratio of \HM is 1, which is as good as \BF and \FF. The expected waste of \HM is $\Theta(\sqrt{n} \log ^{3/4} n)$ which is as good as \BF and better than \FF. The algorithm is easy to implement and has the same running time as \BF.

\HM can be seen as a general way to improve the performance of the Super Harmonic class of algorithms in general and Harmonic algorithm in particular. 
 We illustrate this for the simplest member of this family, namely Refined Harmonic algorithm. To do so, we introduce a new algorithm, called Refined Harmonic Match (\RHM), and show that the competitive ratio of the algorithm is at most equal to 1.636 of Refined Harmonic, while its average-case ratio is 1 which as good as \BF and \HM. The expected waste of \RHM is equal to that of \BF. Consequently, the algorithm achieves the desired average-case performance of \BF and also the worst-case performance of Refined Harmonic. 

Similar to \HA, \HM and \RHM are based on classifying items based on their sizes and treating items of each class (almost) separately. To boost the average-case performance, these algorithms \textit{match} large items with proportionally smaller items through assigning them to the same classes. Careful definition of classes results in the same average-case performance of \MBF. To some extent, our competitive analyses of \HM and \RHM are similar to those of \HA and Refined \HA, respectively. %Hence, our worst-case upper bound arguments are simple and easy-to-follow. 
Similarly, the average-case analyses of the algorithms are closely related to the analysis of \MBF algorithm and uses similar techniques.

%Many online bin packing algorithms, e.g., \BF and \FF, have the undesired property that removing an item might increase the cost of the solution found by the algorithm \cite{Murgolo88}. This `anomalous' behavior results in an unstable algorithm which is harder to analyze. We show that \HM is a \textit{monotone} algorithm, i.e., removing an item does not increase its cost.\footnote{A stronger notion of anomaly is studied in \cite{Murgolo88} where an algorithm is anomalous if decreasing the size of an item increases its cost. An algorithm might be anomalous with that definition, while being monotone under our definition of anomaly.} This property might have practical significance while at the same time facilitates the theoretical analysis of the algorithm (e.g., we use monotony of various algorithms in our analysis in this paper).

To evaluate the average-case performance of the introduced algorithms in real-world scenarios, we tested them on sequences that follow \emph{discrete} uniform distribution as well as other distributions. 
We compared \HM and \RHM with the existing algorithms and observed that they have comparable performance with \BF and \FF. At the same time, these algorithms had a considerable advantage over other members of the Harmonic family of algorithms. We conclude that \HM and \RHM have better average-case performance than existing algorithms which outperform \BF and \FF in the worst-case scenarios.

%An algorithm is called monotone if it is not anomalous.\footnote{A stronger notion of anomaly is studied in \cite{Murgolo88} where an algorithm is anomalous if decreasing the size of an item increases its cost. An algorithm might be anomalous with that definition, while being monotone under our definition of anomaly.} Many online algorithms, in particular \BF and \FF, are known to be anomalous \cite{Murgolo88}. On the other hand \NF is monotone \cite{Murgolo88}. Also, \HA is monotone and removing an item does not increase the cost of a solution. (Decreasing an item size might increase the cost for \HA.) We show that \HM is also monotone. This property might have practical significance while at the same time faciliatates the theoretical analysis of the algorithm (e.g., we use monotonicity of various algorithms in our analysis in this paper) while at the same time it.

%Due to space restrictions, many proofs have been removed, but they can all be found in the appendix.

\begin{small}
\begin{table*}[!t]

\begin{center}
\begin{tabular}{|c|c|c|c|| }
\hline
  	 Algorithm 																&  Average Ratio 						& Expected waste		& Competitive Ratio														\\
%***********************************************************************************************************************************************************	 
	\hline
	\hline
	 Next Fit (\NF) 										& $1.\bar{3}$ \cite{CoHoSY80} 							& $\Omega(n)$																							& 2 																					\\
	  \hline
	 Best Fit (\BF) 										& 1 \cite{BeJLMM84}													& $\Theta(\sqrt{n} \log ^{3/4} n)$ \cite{Shor86,LeiSho89} & 1.7 \cite{JoDUGG74} 													\\
	  \hline	
	 First Fit (\FF) 										& 1	\cite{LeiSho89}						  						& $\Theta(n^{2/3})$ \cite{Shor86,CoJoSW95}.								& 1.7 \cite{JoDUGG74}      		  								\\
	  \hline
	 Harmonic (\HA)     								& 1.2899 \cite{LeeLee85}										& $\Omega(n)$ 																						& $\rightarrow T_\infty \approx 1.691$ \cite{LeeLee85}					  						\\
	 \hline
	 Refined First Fit (\RFF)						&	$> 1$																			& $\Omega(n)$																							& $1.\bar{6}$ \cite{Yao80A}											\\
	 \hline
	 Refined Harmonic (\RH) 						& 1.2824 \cite{GuChXu02}										& $\Omega(n)$																							& 1.636\cite{LeeLee85,GuChXu02} 							\\
	 \hline
	 Modified Harmonic (\MH) 						&	1.189 \cite{RamaTsuga89}									& $\Omega(n)$																							& 1.615\cite{RamBrowLeeLee89} 								\\
	 \hline	
	 Harmonic++  												&	$> 1$																			& $\Omega(n)$																							& 1.588 \cite{Seid02}													\\	 
	 \hline	
	 \textbf{Harmonic Match (\RH)}  		&	$\mathbf{1}$															  & $\mathbf{\Theta(\sqrt{n} \log ^{3/4} n)}$							& $\mathbf{\rightarrow T_\infty \approx 1.691}$ 		  	\\	  \hline 
	\textbf{Refined Harmonic Match (\RHM)}  				&	$\mathbf{1}$															  & $\mathbf{\Theta(\sqrt{n} \log ^{3/4} n)}$	& $\mathbf{< 1.636}$ 		  	\\	 
	%\Opt																			& $1$																				& $\Theta(\sqrt{n})$	\cite{Knoede81,Lueker82}						& 1  																	\\
	
   \hline \end{tabular}
  \caption{\label{competitiveTable} Average performance ratio, expected waste (under continuous uniform distribution), and competitive ratios for different bin packing algorithms. Results in bold are our contributions.}

\end{center}

\end{table*}

\end{small}

%%%%%%%%%%%%%%%%%%%%%%%%%%%%%%%%%%%%%%%%%%%%%%%%%%%%%%%%%%%%%%%%%%%%%%%%%%%%%%%%%%%%%%%%%%%%%%%%%%%%%%%%%%%%%%%%%%%%%%%%%%%%%%%%%%%%%%%%%%%%%%%%%%%%%%%%%%%%%%%%%%%%%
%%%%%%%%%%%%%%%%%%%%%%%%%%%%%%%%%%%%%%%%%%%%%%%%%%%%%%%%%%%%%%%%%%%%%%%%%%%%%%%%%%%%%%%%%%%%%%%%%%%%%%%%%%%%%%%%%%%%%%%%%%%%%%%%%%%%%%%%%%%%%%%%%%%%%%%%%%%%%%%%%%%%%
%%%%%%%%%%%%%%%%%%%%%%%%%%%%%%%%%%%%%%%%%%%%%%%%%%%%%%%%%%%%%%%%%%%%%%%%%%%%%%%%%%%%%%%%%%%%%%%%%%%%%%%%%%%%%%%%%%%%%%%%%%%%%%%%%%%%%%%%%%%%%%%%%%%%%%%%%%%%%%%%%%%%%
%%%%%%%%%%%%%%%%%%%%%%%%%%%%%%%%%%%%%%%%%%%%%%%%%%%%%%%%%%%%%%%%%%%%%%%%%%%%%%%%%%%%%%%%%%%%%%%%%%%%%%%%%%%%%%%%%%%%%%%%%%%%%%%%%%%%%%%%%%%%%%%%%%%%%%%%%%%%%%%%%%%%%
%%%%%%%%%%%%%%%%%%%%%%%%%%%%%%%%%%%%%%%%%%%%%%%%%%%%%%%%%%%%%%%%%%%%%%%%%%%%%%%%%%%%%%%%%%%%%%%%%%%%%%%%%%%%%%%%%%%%%%%%%%%%%%%%%%%%%%%%%%%%%%%%%%%%%%%%%%%%%%%%%%%%%
%%%%%%%%%%%%%%%%%%%%%%%%%%%%%%%%%%%%%%%%%%%%%%%%%%%%%%%%%%%%%%%%%%%%%%%%%%%%%%%%%%%%%%%%%%%%%%%%%%%%%%%%%%%%%%%%%%%%%%%%%%%%%%%%%%%%%%%%%%%%%%%%%%%%%%%%%%%%%%%%%%%%%
%%%%%%%%%%%%%%%%%%%%%%%%%%%%%%%%%%%%%%%%%%%%%%%%%%%%%%%%%%%%%%%%%%%%%%%%%%%%%%%%%%%%%%%%%%%%%%%%%%%%%%%%%%%%%%%%%%%%%%%%%%%%%%%%%%%%%%%%%%%%%%%%%%%%%%%%%%%%%%%%%%%%%
%%%%%%%%%%%%%%%%%%%%%%%%%%%%%%%%%%%%%%%%%%%%%%%%%%%%%%%%%%%%%%%%%%%%%%%%%%%%%%%%%%%%%%%%%%%%%%%%%%%%%%%%%%%%%%%%%%%%%%%%%%%%%%%%%%%%%%%%%%%%%%%%%%%%%%%%%%%%%%%%%%%%%

\section{Harmonic Match Algorithm}
Recall that Harmonic (\HA) algorithm which has a parameter $K$ and defines $K$ classes $(1/2,1], (1/3,1/2],$ $\ldots, (1/(K-1),1/K]$, and $(0,1/K]$; items in the same class are separately treated using the \algo{Next Fit} strategy. 
Similarly to Harmonic algorithm, Harmonic Match has a parameter $K$ and divides items into $K$ classes based on their sizes. %Items which belong to the same class are treated apart from items in other classes. 
We use $\HM_K$ to refer to Harmonic Match with parameter $K$. The algorithm defines $K$ pairs of intervals as follows. The $i$th pair ($1 \leq i \leq k-1 $) contains intervals $(\frac{1}{i+2},\frac{1}{i+1}]$ and $(\frac{i}{i+1},\frac{i+1}{i+2}]$. The $K$th pair includes intervals $(0,\frac{1}{K+1}]$ and $(\frac{K}{K+1},1]$. An item $x$ belongs to class $i$ if the size of $x$ lies in any of the two intervals associated with the $i$th pair (see Figure \ref{fig:intervals}). Intuitively, the items which are `very large' or `very small' belong to the $K$th class, and as the item sizes become more moderate, they belong to classes with smaller indices. 

\begin{figure}[!t]
\centering
\includegraphics[width=0.4\columnwidth, trim = 0mm 202mm 94mm 0mm, clip]{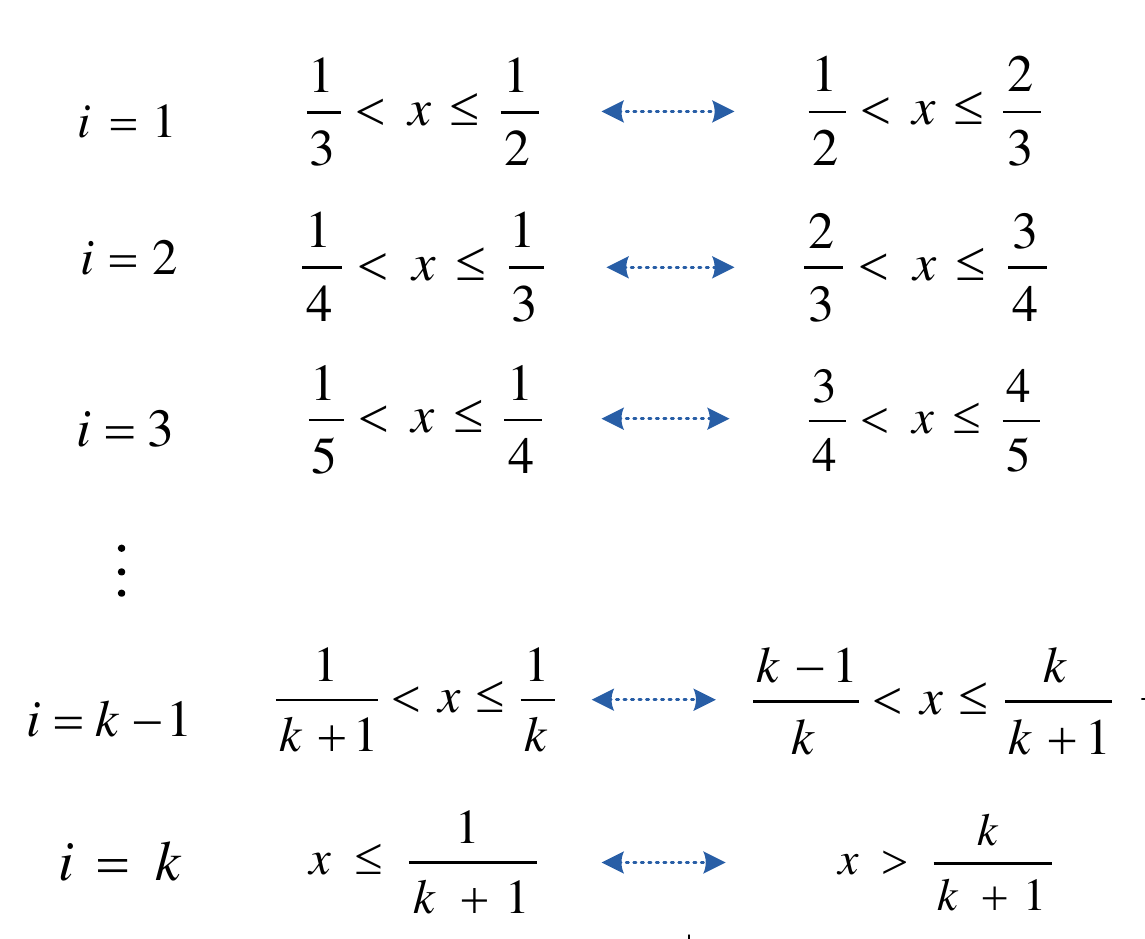}
\caption{The classes defined by \HM. The algorithm matches items from intervals indicated by arrows.}
\label{fig:intervals}
\end{figure}

%\begin{SCfigure}
%\caption{The classes defined by \HM. A bin in the final packing of the algorithm includes items from intervals indicated by arrows.}
%\hspace{1cm}
%\includegraphics[width=0.4\columnwidth, trim = 0mm 202mm 94mm 0mm, clip]{figures/1.pdf}
%\label{fig:intervals}
%\end{SCfigure}

When compared to the intervals of Harmonic algorithms, one can see the first interval of the $i$th pair in the Harmonic Match algorithm $\HM_K$ is the same as the $(i+1)$th interval of Harmonic algorithm $\HA_{K+1}$ $(1 \leq i\leq K)$. Namely, the intervals are the same in both algorithms except that the interval $(\frac{1}{2},1]$ of \HA is further divided into $K+1$ more intervals. In other words, $\HM_K$ is similar to $\HA_{K+1}$, except that it tries to match large items with proportionally smaller item. 
The pairs of intervals which define a class in \HM have the same length, e.g., in the first pair, both intervals have length $\frac{1}{6}$. This property is essential for having good average-case performance.

The packing maintained by \HM include two types of bins: the \emph{mature} bins which are almost full and \emph{normal} bins might become mature by receiving more items. 
For placing an item $x$, \HM detects the class that $x$ belongs to and applies the following strategy to place $x$. 
If $x$ is large item (recall that by large we mean larger than $\frac{1}{2}$), \HM opens a new bin for $x$ and declares it as a normal bin. If $x$ is small, the algorithm applies \BF strategy to place $x$ in a mature bin. If there is no mature bin with enough space, the \BF strategy is applied again to place $x$ in a normal bin which contains the largest `companion' of $x$. A companion of $x$ is a large item of the same class which fits with $x$ in the same bin. In case the \BF strategy succeeds to place $x$ in a bin (i.e., there is a normal bin with a companion of $x$) the selected bin is declared as being mature. Otherwise (when there is no companion for $x$), the algorithm applies \NF strategy to place $x$ in a single normal bin maintained for that class; such a bin only includes small items of the class. If the bin maintained by the \NF strategy does not have enough space, it is declared as a mature bin and a new \NF-bin is opened for $x$. Note that \HM, as defined above, is simple to implement and its time complexity is as good as \BF.

\HM treats items of the same class in a similar way that \algo{Online Match (\OM)} algorithm does, except that there is no restriction on the sum of the sizes of two companion items. Recall that \OM has a parameter which defines a lower bound for the sum of two items in a bin. To facilitate our analysis in the following sections, we define algorithm \algo{Relaxed Online Match (\ROM)} as a subroutin of \HM as follows. To place a large item, \ROM opens a new bin. To place a small item $x$, it applies the \BF strategy to place $x$ in an open bin with a single large item and closes the bin. If such a bin does not exists, \ROM places $x$ using \NF strategy (and opens a new bin if necessary). Using \ROM, we can describe Harmonic Match algorithm in the following way. To place a small item, $\HM_K$ tries to place it in a mature bin using \BF strategy. Large items and the small items which do not fit in mature bins are treated using \ROM strategy along with other items of their classes (which did not fit in mature bins). The bins which are closed by \ROM strategy are declared as mature bins.

%%%%%%%%%%%%%%%%%%%%%%%%%%%%%%%%%%%%%%%%%%%%%%%%%%%%%%%%%%%%%%%%%%%%%%%%%%%%%%%%%%%%%%%%%%%%%%%%%%%%%%%%%%%%%%%%%%%%%%%%%%%%%%%%%%%%%%%%%%%%%%%%%%%%%%%%%%%%%%%%%%%%%
%%%%%%%%%%%%%%%%%%%%%%%%%%%%%%%%%%%%%%%%%%%%%%%%%%%%%%%%%%%%%%%%%%%%%%%%%%%%%%%%%%%%%%%%%%%%%%%%%%%%%%%%%%%%%%%%%%%%%%%%%%%%%%%%%%%%%%%%%%%%%%%%%%%%%%%%%%%%%%%%%%%%%
%%%%%%%%%%%%%%%%%%%%%%%%%%%%%%%%%%%%%%%%%%%%%%%%%%%%%%%%%%%%%%%%%%%%%%%%%%%%%%%%%%%%%%%%%%%%%%%%%%%%%%%%%%%%%%%%%%%%%%%%%%%%%%%%%%%%%%%%%%%%%%%%%%%%%%%%%%%%%%%%%%%%%
%%%%%%%%%%%%%%%%%%%%%%%%%%%%%%%%%%%%%%%%%%%%%%%%%%%%%%%%%%%%%%%%%%%%%%%%%%%%%%%%%%%%%%%%%%%%%%%%%%%%%%%%%%%%%%%%%%%%%%%%%%%%%%%%%%%%%%%%%%%%%%%%%%%%%%%%%%%%%%%%%%%%%
%%%%%%%%%%%%%%%%%%%%%%%%%%%%%%%%%%%%%%%%%%%%%%%%%%%%%%%%%%%%%%%%%%%%%%%%%%%%%%%%%%%%%%%%%%%%%%%%%%%%%%%%%%%%%%%%%%%%%%%%%%%%%%%%%%%%%%%%%%%%%%%%%%%%%%%%%%%%%%%%%%%%%
%%%%%%%%%%%%%%%%%%%%%%%%%%%%%%%%%%%%%%%%%%%%%%%%%%%%%%%%%%%%%%%%%%%%%%%%%%%%%%%%%%%%%%%%%%%%%%%%%%%%%%%%%%%%%%%%%%%%%%%%%%%%%%%%%%%%%%%%%%%%%%%%%%%%%%%%%%%%%%%%%%%%%
%%%%%%%%%%%%%%%%%%%%%%%%%%%%%%%%%%%%%%%%%%%%%%%%%%%%%%%%%%%%%%%%%%%%%%%%%%%%%%%%%%%%%%%%%%%%%%%%%%%%%%%%%%%%%%%%%%%%%%%%%%%%%%%%%%%%%%%%%%%%%%%%%%%%%%%%%%%%%%%%%%%%%
%%%%%%%%%%%%%%%%%%%%%%%%%%%%%%%%%%%%%%%%%%%%%%%%%%%%%%%%%%%%%%%%%%%%%%%%%%%%%%%%%%%%%%%%%%%%%%%%%%%%%%%%%%%%%%%%%%%%%%%%%%%%%%%%%%%%%%%%%%%%%%%%%%%%%%%%%%%%%%%%%%%%%

\subsection{Worst-Case Analysis}

For the worst-case analysis of \HM, we observe that the Harmonic algorithm is \emph{monotone} in the sense that removing an item does not increase its cost:

\begin{lemma}\label{HAMono}
Removing an item does not increase the costs for the Harmonic algorithm.
\end{lemma}

\begin{proof}
Recall that $\HA_K$ defines a class for each item and applies the \NF strategy to place each item together with items of the same class. So, the cost of the algorithm for serving a sequence $\sigma$ is $\NF(\sigma_1) + \NF(\sigma_2) + \ldots + \NF(\sigma_K) $, where $\sigma_i$ is the sequence of items which belong to class $i$. Assume an item $x$ is removed from $\sigma$ and let $j$ denote the class that $x$ belongs to $(1 \leq j \leq K)$. The cost of $\HA_K$ for serving the reduced sequence (in which $x$ is removed) will be the same except that $\NF(\sigma_j)$ is replaced by $\NF(\sigma'_j)$, where $\sigma'_j$ is a copy of $\sigma_j$ in which $x$ is missing. Since \NF is monotone \cite{Murgolo88}, we have $\NF(\sigma'_j) \leq \NF(\sigma_j)$. Consequently, the cost of \HA cannot increase after removing $x$. %\qed
\end{proof}

We use the above lemma to show that the cost of $\HM_K$ for serving any sequence $\sigma$ is no larger than that of $\HA_{K+1}$. Consequently, the competitive ratio of $\HM_K$ is no larger than that of $\HA_{K+1}$.

\begin{theorem} \label{th1}
The cost of $\HM_{K}$ to serve any sequence $\sigma$ is no larger than that of $\HA_{K+1}$.
\end{theorem}

\begin{proof}
Consider the final packing of \HM for serving $\sigma$. Colour a small item \textit{red} if %it has a companion, i.e., 
it is packed with a large item in the same bin; colour all other small items \textit{white}. 
Consider the sequence $\sigma$' which is the same as $\sigma$ except that the red items are removed. We claim $\HM_{K}(\sigma)$ = $\HA_{K+1}(\sigma')$. Let $\sigma_i$ denote the sequence of items which belong to class $i$ of $\HM_{K}$ ($1 \leq i \leq K$). The cost of $\HM_K$ for serving $\sigma_i$ is $l_i + \NF(W_i)$, where $l_i$ is the number of large items $\sigma_i$ and $W_i$ is the sequence formed by white items in $\sigma_i$. Let $\sigma'_i$ be a subsequence of $\sigma_i$ in which red items are removed (hence it is also a subsequence of $\sigma'$). 
Since small and large items are treated separately by $\HA_{K+1}$, the cost of $\HA_{K+1}$ for serving $\sigma'_i$ is also $l_i + \NF(W_i)$. Hence, $\HM_K(\sigma_i) = \HA_{K+1}(\sigma'_i)$. Taking the sum over all classes, we get $\HM_{K}(\sigma)$ = $\HA_{K+1}(\sigma')$.
On the other hand, by Lemma \ref{HAMono}, \HA is monotone and $\HA_{K+1}(\sigma') \leq \HA_{K+1}(\sigma)$. Consequently, $\HM_{K}(\sigma) \leq \HA_{K+1}(\sigma)$. %\qed
\end{proof}

We show that the upper bound given in the above theorem is tight. Consequently, we have:

\begin{corollary} \label{col1}
The competitive ratio of $\HM_K$ is equal to that of $\HA_{K+1}$, i.e., it converges to $T_\infty \approx 1.691$ for large values of $K$. 
\end{corollary}

\begin{proof}
%Regarding the upper bound, by Theorem \ref{th1}, the competitive ratio of $\HM_K$ cannot be more than that of $\HA_{K+1}$. 
Let $\alpha$ denote a lower bound for the competitive ratio of $\HA_{K+1}$ and consider a sequence $\sigma$ for which the cost of $\HA_{K+1}$ is $\alpha$ times more than that of \Opt. 
Define a sequence $\sigma_\pi$ as a permutation of $\sigma$ in which items are sorted in increasing order of their sizes. When applying $\HM_K$ on $\sigma_\pi$, all large items will be unmatched in their bins (no other item is packed in their bins). Hence, the cost of $\HM_K$ for packing $\sigma_\pi$ is the same as $\HA_{K+1}$ for packing $\sigma$, i.e., $\alpha$ times the cost of \opt for serving $\sigma$ and $\sigma_\pi$ (note that \opt uses an identical packing for both $\sigma$ and $\sigma_\pi$). % since the offline multiset of items are the same). %\qed
\end{proof}

To achieve a competitive ratio better than 1.7 of \BF and \FF for $\HM_K$, it is sufficient to have $K \geq 6$. In that case, $\HM_6$ performs as well as $\HA_{7}$, which has a competitive ratio of at most 1.695. %For the suggested value of $K=10$, the competitive ratio improves to at most 1.693. 
%It should be mentioned that the existing bounds for the competitive ratio of $\HA_K$ are not tight when $K>7$ \cite{CoGaJo97}.

%%%%%%%%%%%%%%%%%%%%%%%%%%%%%%%%%%%%%%%%%%%%%%%%%%%%%%%%%%%%%%%%%%%%%%%%%%%%%%%%%%%%%%%%%%%%%%%%%%%%%%%%%%%%%%%%%%%%%%%%%%%%%%%%%%%%%%%%%%%%%%%%%%%%%%%%%%%%%%%%%%%%%
%%%%%%%%%%%%%%%%%%%%%%%%%%%%%%%%%%%%%%%%%%%%%%%%%%%%%%%%%%%%%%%%%%%%%%%%%%%%%%%%%%%%%%%%%%%%%%%%%%%%%%%%%%%%%%%%%%%%%%%%%%%%%%%%%%%%%%%%%%%%%%%%%%%%%%%%%%%%%%%%%%%%%
%%%%%%%%%%%%%%%%%%%%%%%%%%%%%%%%%%%%%%%%%%%%%%%%%%%%%%%%%%%%%%%%%%%%%%%%%%%%%%%%%%%%%%%%%%%%%%%%%%%%%%%%%%%%%%%%%%%%%%%%%%%%%%%%%%%%%%%%%%%%%%%%%%%%%%%%%%%%%%%%%%%%%
%%%%%%%%%%%%%%%%%%%%%%%%%%%%%%%%%%%%%%%%%%%%%%%%%%%%%%%%%%%%%%%%%%%%%%%%%%%%%%%%%%%%%%%%%%%%%%%%%%%%%%%%%%%%%%%%%%%%%%%%%%%%%%%%%%%%%%%%%%%%%%%%%%%%%%%%%%%%%%%%%%%%%
%%%%%%%%%%%%%%%%%%%%%%%%%%%%%%%%%%%%%%%%%%%%%%%%%%%%%%%%%%%%%%%%%%%%%%%%%%%%%%%%%%%%%%%%%%%%%%%%%%%%%%%%%%%%%%%%%%%%%%%%%%%%%%%%%%%%%%%%%%%%%%%%%%%%%%%%%%%%%%%%%%%%%
%%%%%%%%%%%%%%%%%%%%%%%%%%%%%%%%%%%%%%%%%%%%%%%%%%%%%%%%%%%%%%%%%%%%%%%%%%%%%%%%%%%%%%%%%%%%%%%%%%%%%%%%%%%%%%%%%%%%%%%%%%%%%%%%%%%%%%%%%%%%%%%%%%%%%%%%%%%%%%%%%%%%%
%%%%%%%%%%%%%%%%%%%%%%%%%%%%%%%%%%%%%%%%%%%%%%%%%%%%%%%%%%%%%%%%%%%%%%%%%%%%%%%%%%%%%%%%%%%%%%%%%%%%%%%%%%%%%%%%%%%%%%%%%%%%%%%%%%%%%%%%%%%%%%%%%%%%%%%%%%%%%%%%%%%%%
%%%%%%%%%%%%%%%%%%%%%%%%%%%%%%%%%%%%%%%%%%%%%%%%%%%%%%%%%%%%%%%%%%%%%%%%%%%%%%%%%%%%%%%%%%%%%%%%%%%%%%%%%%%%%%%%%%%%%%%%%%%%%%%%%%%%%%%%%%%%%%%%%%%%%%%%%%%%%%%%%%%%%

\subsection{Average-Case Analysis}\label{avgHmSection}
In this section, we study the average-case performance of the \HM algorithm under a uniform distribution. % as is standard in the field. 
Like most related work, we make use of the results related to the \textit{up-right matching} problem. An instance of this problem includes $n$ points generated uniformly at random in a unit-square in the plane. Each point receives a $\oplus$ or $\ominus$ label with an equal probability. The goal is to find a maximum matching of $\oplus$ points with $\ominus$ points so that in each pair of matched points the $\oplus$ point appear above and to the right of the $\ominus$ point. Let $U_n$ denote the number of unmatched points in an optimal up-right matching of $n$ points. For the expected size of $U_n$, it is known that $E[U_n] = \Theta (\sqrt{n} \log ^{3/4}n)$ \cite{Shor86,LeiSho89,RheTal88a,CofSho91}. Given an instance of bin packing defined by a sequence $\sigma$, one can make an instance of up-right matching as follows \cite{KaLuMa84}: Each item $x$ of $\sigma$ is plotted as a point in the unit square. 
the vertical coordinate of such point corresponds to the index of $x$ in $\sigma$ (normalized to fit in the square). If $x$ is smaller than $1/2$, the point associated with $x$ is labeled as $\ominus$ and its horizontal coordinate will be $2x$; otherwise, the point will be $\oplus$ and its horizontal coordinate will be $2-2x$. Note that the resulted point will be bounded in the unit square. 
A solution to the up-right matching instance gives a packing of $\sigma$ in which the items associated with a pair of matched points are placed in the same bin. Note that the sum of the sizes of these two items is no more than the bin capacity. Also, in such solution, each bin contains at most two items.

For our purposes, we study $\sigma_t$ as a subsequence of $\sigma$ which only includes items which belong to the same class in the \HM algorithm. The items in $\sigma_t$ are generated uniformly at random from $(\frac{1}{t+1},\frac{1}{t}] \cup (\frac{t-1}{t},\frac{t}{t+1}]$. Since the two intervals have the same length, the items can be plotted in a similar manner on the unit square as follows. %To be more precise, t
The horizontal coordinate of a small item with size $x$ is $x \times t(t+1) - t$ and for large items it is $x \times t(t+1) - (t^2-1)$. The label of the item and its vertical coordinate are defined as before.

Any bin packing algorithm which closes a bin after placing a small item can be applied to the up-right matching problem. Each edge in the up-right matching instance corresponds to a bin which includes one small and one large item. Recall that the algorithm Matching Best Fit (\MBF) applies a \algo{Best Fit} strategy except that it closes a bin as soon as it receives an item with size smaller than or equal to 1/2. So, \MBF can be applied in the up-right matching problem. Indeed, it creates an optimal up-right matching, i.e., if we apply \MBF on a sequence $\sigma_t$ which is randomly generated from $(0,1]$, the number of unmatched points will be $\Theta (\sqrt{n_t} \log ^{3/4}n_t)$, where $n_t$ is the length of $\sigma_t$ \cite{Shor86}. %Note that the matching algorithm associated with \MBF is independent of the distribution, and 
We show the same result holds for the bin packing sequences in which items are taken uniformly at random from $(\frac{1}{t+1},\frac{1}{t}] \cup (\frac{t-1}{t},\frac{t}{t+1}]$:

\begin{lemma}\label{MBFEx}
For a sequence $\sigma_t$ of length $n_t$ in which item sizes are selected uniformly at random from $(\frac{1}{t+1},\frac{1}{t}] \cup (\frac{t-1}{t},\frac{t}{t+1}]$, we have $E[\MBF(\sigma_t)] = n_t/2 + \Theta(\sqrt{n_t} \log^{3/4}n_t)$.
\end{lemma}

\begin{proof}
Define an instance for up-right matching from $\sigma_t$ as follows: Let $x$ be the $i$th item of $\sigma_t$ $(1\leq i \leq n_t)$. If $x$ is small, plot a point with $\ominus$ label at position $ (x \times t(t+1) - t, t/n)$; otherwise, plot a point with $\oplus$ label at position $(x \times t(t+1) - (t^2-1),t/n)$. This way, the points will be bounded in the unit square. Since the item sizes are generated uniformly at random from the two intervals and the sizes of the intervals are the same, the point locations and labels are assigned uniformly at random. As a result, the number of unmatched points in the up-right matching solution by \MBF is expected to be $\Theta(\sqrt{n_t} \log^{3/4}n_t)$. The unmatched points are associated with the items in $\sigma_t$ which are packed as a single item in their bins by \MBF. Let $s$ denote the number of such items, hence $E[s] \in \Theta(\sqrt{n_t} \log^{3/4}n_t)$. Except these $s$ items, other items are packed with exactly one other item in the same bin. So we have $\MBF(\sigma_t) - s/2 = n_t/2$ which implies $E[\MBF(\sigma_t)] = n_t/2 + E[s]/2$. Since $E[s] \in \Theta(\sqrt{n_t} \log^{3/4}n_t)$, the statement of the lemma follows. %\qed 
\end{proof}

\begin{lemma} \label{OMMBF}
For any instance $\sigma$ of the bin packing problem, the cost of \ROM for serving $\sigma$ is no more than that of \MBF.
\end{lemma}

\begin{proof}
Both \ROM and \MBF open a new bin for each large item. Also, they treat small items which have companions in the same way, i.e., they place the item in the bin of the largest companion and close that bin. The only difference between \ROM and \MBF is in placing small items without companions where \ROM applies the \NF strategy while \MBF opens a new bin for each item. Trivially, \ROM does not open more bins than \MBF for these items. 
\end{proof}

\begin{lemma}\label{MBFMONO}
Removing an item does not increase the cost of \MBF.
\end{lemma}

\begin{proof}
Let $\sigma$ denote an input sequence and $n$ denote the length of $\sigma$. We use a reverse induction to show that removing the $(n-i)$th item ($0\leq i \leq n-1$) does not increase the cost of \MBF for serving $\sigma$. %, i.e., $\ROM(\sigma^{-x}) \leq \ROM(\sigma)$. 
Note that removing the last item does not increase the cost of any algorithm and the base of induction holds. Assume the statement holds for $i = k+1$, i.e., removing any item from index $i \geq k+1$ does not increase the cost of \MBF. We show the same holds for $i=k$. Let $n_l$ denote the number of large items in sigma and $n_{ss}$ denote the number of \textit{single small} items, which are the small items which have no companion. For placing a single small item, \MBF opens a bin and closes the bin right after placing the item. 
For the cost of \MBF for serving $\sigma$ we have $\MBF(\sigma) = n_l + n_{ss}$.
Let $x$ denote the $k$th item in $\sigma$. We show removing $x$ does not increase the cost of \MBF. There are a few cases to consider. 

First, note that if $x$ is a small single item, removing it decreases the cost of \MBF by one unit. Since the packing of other items does not change, the inductive step trivially holds.
Next, assume $x$ is a small item which has a companion. Removing $x$ might create a space for another small item $x'$ in the bin of $x$. In case such an item does not exist (i.e., no other item replaces $x$ in its bin), the packing and consequently the cost of the algorithm do not change. Otherwise, $x'$ is placed in the bin which includes the companion of $x$ and closes that bin. Let $k'$ denote the index of $x'$ in the sequence and note than $k'>k$. Also, let $\sigma^{-a}$ denote a copy of $\sigma$ from which an item $a$ is removed. We have $\MBF(\sigma^{-x}) = \MBF(\sigma^{-x'}$), i.e., 
removing item $x$ changes the cost of \MBF in the same way that removing $x'$ does. By the induction hypothesis, removing $x'$ does not increase the cost of \MBF and we are done.

The only remaining cases is when $x$ is a large item. If $x$ does not have a companion, removing $x$ decreases the cost by one unit and we are done. Now assume $x$ has a companion $x'$ and let $\sigma^{--}$ denote the same sequence as $\sigma$ in which both $x$ and $x'$ are removed. As before, let $\sigma^{-x}$ denote a copy of $\sigma$ in which $x$ is removed. We have $\MBF(\sigma^{--}) = \MBF(\sigma) -1$. On the other hand, $\MBF(\sigma^{-x}) \leq \MBF(\sigma^{--}) + 1$. This is because adding a small item to a sequence does not increase the cost of \MBF by more than one unit; this holds because \MBF closes a bin as soon as a small item is placed in the bin. We conclude that $\MBF(\sigma^{-x}) \leq \MBF(\sigma)$, and the inductive step holds.

\end{proof}

Provided with the above lemmas, we prove the following theorem.

\begin{theorem} \label{majorAvg}
Let $\sigma$ be a sequence of length $n$ in which item sizes are selected uniformly at random from $(0,1]$. The expected wasted space of \HM for packing $\sigma$ is $\Theta(\sqrt{n} \log ^{3/4} n)$.
\end{theorem}

\begin{proof}

Let $\sigma^-$ be a copy of $\sigma$ in which those items which are placed in mature bins are removed. % (these are small items that are placed by \HM in bins which have been mature before placing them).
Let $\sigma_1^-, \ldots, \sigma_K^-$ be the subsequences of $\sigma^-$ formed by items belonging to different classes of \HM. We have:
\begin{equation*}
\HM(\sigma) = \sum\limits_{t=1}^{K} \ROM(\sigma_t^-) \leq \sum\limits_{t=1}^{K} \MBF(\sigma_t^-) \leq \sum\limits_{t=1}^{K} \MBF(\sigma_t)
\label{eq:1}
\end{equation*}
%and by Lemma \ref{OMMBF} we get 
%\begin{equation*}
%\HM(\sigma) \leq \sum\limits_{t=1}^{K} \MBF(\sigma_t)
%\label{eq:2}
%\end{equation*}
The inequalities come from Lemmas \ref{OMMBF} and \ref{MBFMONO}, respectively. Consequently, by Lemma \ref{MBFEx}, we have:
\begin{equation*}
E[\HM(\sigma)] \leq \sum\limits_{t=1}^{K} \left(n_t/2 + \Theta(\sqrt{n_t} \log^{3/4}n_t) \right) = \frac{n}{2} + \Theta(\sqrt{n}\log^{3/4}n)
\label{eq:}
\end{equation*}

Note that the last equation only holds when $K$ is a constant. The expected cost of \OPT is no better than $n/2$ since half of items are expected to be larger than 1/2. Consequently, we have: 
\begin{equation*}
E[\HM(\sigma) - \OPT(\sigma)] \in \Theta(\sqrt{n}\log^{3/4}n)
\label{eq:4}
\end{equation*}
%which completes the proof. %\qed
\end{proof}

%Note that, since the expected cost of \OM is $n/2 + o(n)$, the average performance ratio of $\HM_K$ is 1. %This is because the expected cost of \opt is $n/2$. 
It should be mentioned that, although the expected waste of \HM algorithms is $\Theta(\sqrt{n}\log^{3/4}n)$, there is a multiplicative constant involved in the expression which is a function of $K$. This implies that the rate of convergence to \BF is slower for larger values of $K$.

\section{Refined Harmonic Match}

In this section, we introduce a slightly more complicated algorithm, called \algo{Refined Harmonic Match} (\RHM), which has a better competitive ratio than \BF, \FF, and \HM while performing as well as them on average. %, i.e., it has an average-case performance ratio of 1. 
In introducing \RHM, we have been inspired by the Refined Harmonic algorithm of \cite{LeeLee85}. 

Similar to \HM, \RHM divides items into a constant number of classes and treat items of each class separately. The classes defined for \RHM are the same as those of $\HM_K$ with $K=19$. 
The items which belong to class $k \geq 2$ are treated using \HM strategy, i.e., a set of mature bins are maintained. If an item fits in mature bins, it is placed there using \BF strategy; otherwise, it is placed together with similar items of its class using \ROM strategy. At the same time, the bins closed by \ROM subroutin are declared as being mature.

%Like \HM, the items which belong to class $k \geq 2$ are treated using the \ROM strategy. 
The only difference between \HM and \RHM in in packing items of class $1$, i.e., items in range $(1/3,2/3]$. \RHM divides items in this range into four groups $a = (1/3,37/96]$, $b= (37/96,1/2]$, $c = (1/2, 59/96]$, and $d = (59/96,2/3]$ (see Figure \ref{fig:intervalsRHM}). %Note that intervals $a$ and $d$ have equal length, and the same holds for intervals $b$ and $c$.
Similar to Refined Harmonic, to handle the bad sequences which result in lower bound of $T_\infty$ for competitive ratios of \HA and \MH, \RHM designates a fraction of bins opened by items of group $a$ to host future $c$ items. Note that the total size of a $c$ item and an $a$ item is no more than 1. However, to ensure a good average-case performance, \RHM should be more elaborate than Refined Harmonic. This is because it cannot treat $b$ apart from other items using a strategy like \NF as Refined Harmonic does. In what follows, we introduce an online algorithm called Refined Relaxed Online Match (\RRM) as a subroutine of \RHM that is specifically used for placing items of class 1.

\begin{figure}[!t]
\centering
\includegraphics[width=0.45\columnwidth, trim = 0mm 188mm 58mm 0mm, clip]{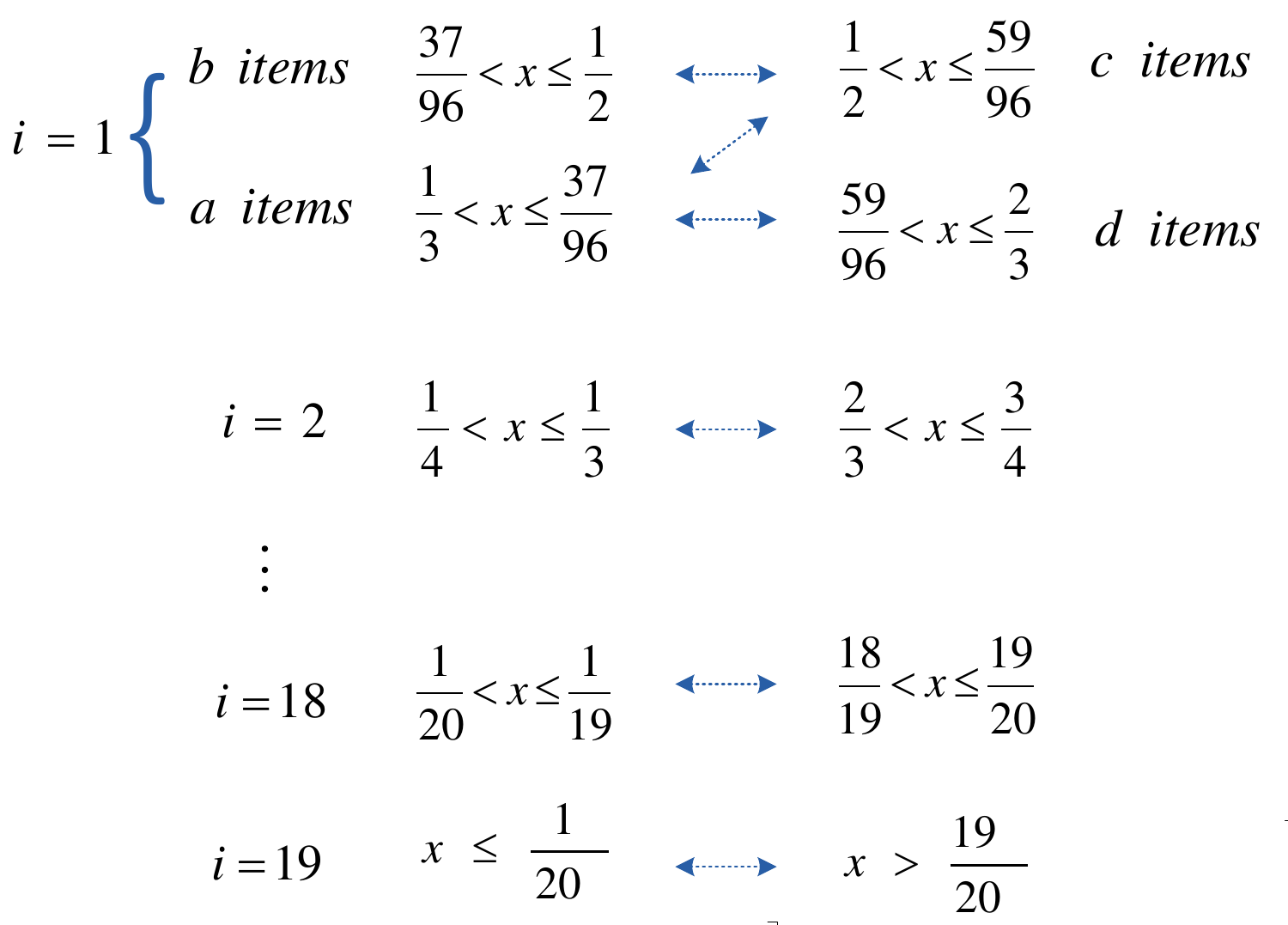}
\caption{The classes defined by \RHM. The algorithm matches items from intervals indicated by arrows.}
\label{fig:intervalsRHM}
\end{figure}

%For placing an item $x \in (1/3,2/3]$, \RRM works as follows. 
To place an item $x$ of the first class ($x \in (1/3,2/3]$), \RRM uses the following strategy. At each step of the algorithm, when two items of the first class are placed in the same bin, that bin is declared as being mature and will be used for placing small items of other classes. More precisely, it will be added to the set of mature bins maintained by the \HM algorithm applied for placing items in other classes.
If $x$ is a $d$-item, \RRM opens a new bin for $x$. If $x$ is a $c$ item, the algorithms checks whether there are bins which include a single $a$ item and are designated to have a $c$ item; in case there are such bins, $x$ is placed in one of those using \BF strategy. %the first bin among those bins (in the order they are opened). 
In case there is no such a bin, a new bin is opened for $x$. 

For $a$ and $b$ items (small items of class 1), \RRM uses the \BF strategy to select a bin with enough space which includes a single large item (if there is such a bin). This is particularly important to guarantee a good average-case behavior. If $x$ is a $b$ item, the algorithms checks the bin with the highest level in which $x$ fits; if such a bin includes a $c$ or a $b$ item, $x$ is placed there. Otherwise (when the selected bin does not exist or when it has an $a$ item), a new bin is opened for $x$. 
If $x$ is an $a$ item, the algorithm uses \BF strategy to place it into a bin with a $d$ or $c$ item. If no suitable bin exist, $x$ is placed into a bin with a single $a$ item (there is at most one such bin); if there is no such bin, a new bin is opened for $x$. 

When a new bin is opened for an $a$-item, the bin will be designated to either include a $c$ item or another $a$ item in the future. 
We define \textit{red bins} as those which include two $a$ items, or a single $a$ item designated to be paired with another $a$ item, and define \textit{blue bins} as those which include either a $c$ item together with an $a$ or a $b$ item, or a single $a$ item designated to be paired with a $c$ in future. When opening a new bin for an $a$ item, \RHM tries to maintain the number of red bins as close to three times the number of blue bins as possible. Namely, if the number of red bins is less than 3 times of blue bins, it declares the opened bin as a red bin to host another $a$ item in future; otherwise, the new open bin is declared as a blue bin to host a $c$ item in future. This way, the number of red bins is close to (but never more than) three times of blue bins. Note that, when many $b$ items are placed together with $c$ items, the resulting bins will be blue. In this case, the algorithm does not limit the number of blue bins unless it opens bins for $a$ items. Consequently, the number of red bins can be less than three times of blue bins. Algorithm \ref{algRRM} illustrates how \RRM works. % (see the next section for some notations). 

%\restylealgo{ruled}%\linesnumbered
\begin{center}
\scalebox{1}{
\begin{algorithm}[H]
  \caption{\RRM algorithm: Placing a sequence of items in range (1/3,2/3]\label{algRRM}}
			\SetKwInOut{Input}{input}
			\Input{A sequence $\sigma = \left\langle \sigma_1, \sigma_2, \ldots, \sigma_n \right\rangle$ of items in range (1/3,2/3]}
			\BlankLine
			\Let{$N_{a_1}, N_{a_2}, N_{aa}, N_{ab}, N_{ac}, N_b, N_{bc}$}{0} \\
%			\SetVline
			\For{$i \gets 1 \textrm{ to } n$}{
				\Switch {$\sigma_i$}{
						\uCase {$d$ item:}{
									open a new bin for $\sigma_i$ 
						}
						\uCase {$c$ item:}{
									\uIf{$N_{a_1} >0$}
									{
										%place $\sigma_i$ in the first bin (with respect to opening time) with a single $a$ item 
										Use \BF to place $\sigma_i$ in a bin with an $a$ item 
										$N_{a_1} \gets N_{a_1} -1$; $N_{ac} \gets N_{ac} +1$\\
										\small{\textit{\{$N_{a_1}$ is the number of bins with a single $a$ item which are designated to have a $c$ item\}}}
									}
									\lElse
									{	open a new bin for $\sigma_i$; $N_c \gets N_c + 1$}
	
						}
						\uCase {$b$ item:}{
									select the bins with a $c$ item which have enough capacity for $\sigma_i$\\
									\uIf {there is a selected bin}
												{place $\sigma_i$ into the bin with the highest level among the selected bins;\\
												$N_{bc} \gets N_{bc} + 1$; $N_c \gets N_c -1$}
									\uElseIf {$N_b =1$}
									{place $\sigma_i$ into the bin with a single $b$ item; $N_b \gets 0$}
									\lElse
									{open a new bin for $\sigma_i$; $N_b \gets 1$}
						}
				
					\uCase{$a$ item:}{
							select the bins with a $d$ item which have enough capacity for $\sigma_i$\\
							\uIf {there is such a bin}
							{place $\sigma_i$ into the bin with the highest level among those bins}
							\uElseIf {$N_c>0$}
					 {	place $\sigma_i$ into the bin with the largest $c$ item; $N_{ac} \gets N_{ac}+1$; $N_{c} \gets N_{c}-1$}
							\uElseIf {$N_{a_2}=1$}
							{place $\sigma_i$ into any with a single $a$ item; $N_{a_2} \gets 0$; $N_{aa} \gets N_{aa}+1$}
							\uElse
						 {place $\sigma_i$ in a new (empty) bin \\
							\small{\textit{\{compare the number of red bins with 3 times number of blue bins\}}} \\
							\uIf {$N_{aa} < 3(N_{ac} + N_{a_1} + N_{bc})$}
						 {$N_{a_2} \gets 1$; \small{\textit{\{declare the opened bin as a red bin (an $a_2$-bin) \}\\}} } %\COMMENT{}
							\uElse
							{$N_{a_1} \gets N_{a_1} + 1$; \small{\textit{\{declare the opened bin as a blue bin (an $a_1$-bin) \}}} } %\COMMENT{declare the opened bin as a blue bin ($a$-bin)}
						}
						}
				}
			
}
\end{algorithm}
}
\end{center}

\subsection{Worst-Case Analysis}
In this section, we provide an upper bound of 1.636 for the competitive ratio of \RHM. We start by introducing some notations. Bins in a packing by \RRM can be divided into the following groups: \emph{$d$-bins} which include a $d$-item (might also include an $a$ item), \emph{$c$-bins} (resp. \emph{$b$-bins)} which include a single $c$ (resp. $b$) item, \emph{$a_1$-bins} (resp. \emph{$a_2$-bins}) which include a single $a$ item and are designated to include a $c$ (resp. an $a$) item in future, \emph{$bb$-bins} (resp. \emph{$aa$-bins}) which include two $b$ (resp. $a$) items, 
\emph{$ac$-bins} which include an $a$ item and a $c$ item, and \emph{$bc$-bins} which include a $b$ item and a $c$ item. Note that there is at most one $b$-bin and one $a_2$-bin (otherwise, two of those bins form a $bb$-bin or $aa$-bin, respectively).
We use respectively capital $N$ and lower case $n$ to refer to the number of bins and items: $N_\alpha$ denotes the number of bins of type $\alpha$, e.g., $N_{ac}$ indicates the number of $ac$-bins. Similarly, $N_{red}, N_{blue}$ denote the number of red and blue bins in the packing. Note that $N_{red} = N_{aa} + N_{a_2}$ and $N_{blue} = N_{ac} + N_{bc} + N_{a_1}$. We use $n_x$ to denote the number of items of type $x$ ($x \in \{a,b,c,d \}$). %Moreover,\xxx{I don't quite follow this paragraph} l
Moreover, we use $n_{b_1}$ to denote the number of $b$ items which are packed with a $c$ item ($n_{b_1} = N_{bc}$) and $n_{b_2}$ to denote the number of other $b$ items ($n_{b_1} +n_{b_2} = n_b$). Counting the number of $a$ and $b_1$ items we get $n_a + n_{b_1} = 2N_{aa} + N_{ac} + N_{a_1} + N_{a_2} + N_{bc}$. Since $N_{a_2}\leq 1$, by definition of red and blue bins, we get

\begin{center}
\begin{equation}
n_a+n_{b_1}-1 \leq 2N_{red} + N_{blue} \leq n_a+n_{b_1}
\label{eq:6star}
\end{equation}
\end{center}

We refer to the above equation in a few places in our analysis. 
Since \RHM uses the same strategy as \HM for placing items in classes $k\geq 2$, and \HM never opens more bins tha \HA does (Theorem \ref{th1}), we can prove the following lemma.

\begin{lemma} \label{lemTotCost}
The cost of \RHM for serving items of classes $k\geq 2$ is upper bounded by 
\begin{center}
\begin{equation*}
\RHM(\sigma) \leq \RRM(\sigma_{cl_1}) + \sum\limits_{k=2}^{18} \frac{n_k}{k+1} + 20W'/19 + 20
\end{equation*}
\end{center}
in which $\sigma_{cl_1}$ is the subsequence formed by items of class 1, $n_k$ is the number of items in class $k$, and $W'$ is the total size of items in class 19 (the last class).
\end{lemma}

\begin{proof}
Since \RHM performs similarly to \HM for placing items of class $k\geq 2$, the proof of Theorem \ref{th1} can be applied to state that \RHM does not open more bins than Harmonic algorithm does for placing these items. Harmonic places $k+1$ items of class $k$ in the same bin ($2 \leq k \leq 18$); hence, it opens at most $\frac{n_k}{k+1} +1$ bins for items in such class. The empty space in any bin assigned to items of the last class (class 19) is at most 1/20 since the size of items in this class is no more than 1/20. Hence, the number of opened bins for this class is at most $20W'/19 + 1$. Note that some items in classes $k \geq 2$ might be placed in the mature bins maintained by \HM (including the bins released by \RRM). When comparing with Harmonic algorithm, we can think of these items as being removed from the Harmonic packing. Since the Harmonic is monotone (Lemma \ref{HAMono}), removing these items does not increase the number of bins. Hence, the claimed upper bound still holds. 
\end{proof}

\begin{theorem}\label{mainRHM}
The competitive ratio of \RHM is at most $373/228 < 1.636$.
\end{theorem}

\begin{proof}

\RRM is defined in a way that no $c$-bin and $a_1$-bin can be open at the same time. We consider the following two cases based on the final packing of \RRM for the subsequence formed by items of type 1: 

\begin{itemize}
	\item Case 1: there is no $c$-bin in the final packing, while there is at least one $a_1$-bin in the packing.
	\item Case 2: there is no $a_1$-bin in the final packing.
\end{itemize}

%In both cases, we formulated the cost of \RRM as a function of the number of items in each group (i.e., as a function of $n_a, n_b, n_c,$ and $n_d$). Then we complete the proof through a weighting function. 
We prove the theorem for these two cases separately.

\textbf{Case 1:} Assume there is no $c$-bin in the final packing, while there is at least one $a_1$-bin in the packing. Let $x$ be the last $a$ item for which an $a_1$-bin is opened. We claim that no blue bin is added to the packing after placing $x$. Blue bins are opened by $a$ or $c$ items. A new blue bin cannot be opened by a $c$ item as such a $c$ item should have been placed in one of the existing $a_1$ bins. Also, a new blue bin cannot be opened by an $a$ item since that results in a bin with a single $a$ item%(because $c$ items are placed in $a_1$ bins using \FF strategy)
; this contradicts $x$ being the last item for which an $a_1$ bin is opened. So, the number of blue bins does not increase after placing $x$. At the time of placing $x$, the number of red bins is no less than three times the number of blue bins; otherwise, the bin opened for $x$ would have been declared as a red bin (i.e., an $a_2$-bin). So, we have $-3 \leq N_{red}-3N_{blue} \leq 3$. Using Equation \ref{eq:6star} we get:
\begin{center}
\begin{align*}
%3n_a/7 + 3n_{b_1}/7 -6/7 & \leq N_{red} \leq 3n_a/7 + 3n_{b_1}/7 +3/7 \nonumber \\
%n_a/7 + n_{b_1}/7 -9/7 & \leq N_{blue} \leq n_a/7 + n_{b_1}/7 +8/7 \nonumber
 N_{red} + N_{blue} \leq 4n_a/7 + 4n_{b_1}/7 +11/7 
%\label{eq:12}
\end{align*}
\end{center}

And for the total cost of \RRM we will have:

\begin{center}
\begin{align*}
\RRM(\sigma_{cl_1}) = &  N_d + N_{ac} + N_{bc} + N_b + N_{bb} + N_a + N_{a'} + N_{aa}  \\
& \leq n_d + n_{b_2}/2 + N_{blue} + N_{red}  \\
& \leq n_d + n_{b_2}/2 + 4n_a/7 + 4n_{b_1}/7 + 11/7\\
& \leq  n_d + 4n_b/7 + 4n_a/7 + 11/7 &
\end{align*}
\end{center}

Plugging this into the upper bound given by Lemma \ref{lemTotCost}, we get:

\begin{center}
\begin{equation}
\RHM(\sigma)  \leq \sum\limits_{k=2}^{19} \frac{n_k}{k+1} + 20W'/19 + n_d + 4n_b/7 + 4n_a/7 + 22 
\end{equation}
\label{eq:111}
\end{center}

To further analyze the algorithm, we use a weighting function similar to that of \cite{LeeLee85}. We define a weight for each item so that the total weight of items in a sequence, denoted by $W(\sigma)$, becomes an upper bound for the cost of \RHM. At the same time, we show that the total weight of any set of items which fit in a bin is at most 1.63; this implies that the cost of \opt for serving $\sigma$ is at least $W(\sigma)/1.63$. Consequently, the ratio between the costs of \RHM and \OPT is at most 1.63. %The details of the weighting function can be found in the appendix. 

We define a weight for each item in the following manner. Small items in class $k (2\leq k \leq 18)$ have weight $1/(k+1)$. The weight of a small item $x$ of class $19$ is $20x/19$ . The weight of $d$ items and items larger than 2/3 is 1; the weight of $c$ item is 0, and the weight of $b$ and $a$ items is 4/7. This way, as the above Inequality \ref{eq:111} suggests, the total weight of items in a sequence is an upper bound for the cost of \RHM. Next, we study the maximum weight of items in a bin of \opt. Let $\beta_1,\beta_2, \ldots, \beta_t$ denote the set of items in a bin of \opt so that $\beta_1 \geq \beta_2 \geq \ldots \geq \beta_t$. Let $W_{opt}$ denote the total weight of items in such a bin, i.e., $W_{opt} = \sum\limits_{i=1}^t \omega(\beta_i)$. We claim that $W_{opt} < 1.63$, which implies the competitive ratio is upper bounded by $W_{opt}$. The \textit{density} of any item, defined as the ratio between the weight and the size of the item, is at most $\frac{4/7}{1/3} = 12/7<1.72$ for $b$-items. For $a$ items, the density is at most $\frac{4/7}{37/96} < 1.48$. For items smaller than $b$ items, the density decreases as the class increases (from at most $4/3$ for items of class 2 to at most $20/19$ for items in classes 18 and 19). The density of large items is at most $96/59 < 1.63$. 

To prove the claim, we do a case analysis on the value of $\beta_1$: I) assume $\beta_1$ has a size larger than 59/96, i.e., it is a $d$ item or larger. If $\beta_2$ is an $a$ or a $b$ item, we will have $\beta_3 + \ldots + \beta_t \leq 5/96<1/19$. So, all other items belong to classes 18 or 19 and their density is at most 20/19. Hence, the total weight of items in the bin will be at most $1+4/7+5/96 \times 20/19 < 1.63$. If $\beta_2$ is smaller than or equal to 1/3, the total size of all items except $\beta_1$ is at most 37/96 and their density is at most 4/3; the total weight will be $1 + 37/96 \times 4/3 < 1.52$. II) Assume $\beta_1$ is a $c$ item, i.e., its weight is 0. The total sizes of other items in the bin (all items except $\beta_1$) is at most 1/2 ad their density is upper bounded by 12/7. Hence, the total weight of items in the bin will be $1/2 \times 12/7 <1$. III) Assume $\beta_1$ is a $b$ item. Assume $\beta_2$ is a $b$ or an $a$ item; so, the total size of all other items is at most 27/96, while their density is at most 4/3 (note that they belong to class 2 or higher). Hence, the total weight of items in the bin will be $4/7+4/7+27/96\times 4/3 <1.52$. Next, assume $\beta_2$ is smaller than $a$ items; the total size of all items (except $\beta_1$) wil be at most 59/95 while their weight is at most 4/3. The total weight of items in a bin will be at most $4/7+59/96\times 4/3 < 1.4$. 
IV) Assume $\beta_1$ is an $a$ item. Assume $\beta_2$ is also an $a$ item; so, the total size of all other items is at most 1/3, while having a density of at most 4/3 (note that they belong to class 2 or higher). Hence, the total weight of items in the bin will be $4/7+4/7+1/3\times 4/3 <1.59$. Next, assume $\beta_2$ is smaller than $a$ items; the total size of all items (except $\beta_1$) will be at most 2/3 while their weight is at most 4/3. The total weight of items in a bin will be at most $4/7+2/3\times 4/3 < 1.46$. V) Finally, consider $\beta_1$ is smaller than 1/3; hence, the density of all items is at most 4/3. Consequently, the total weight of items in the bin is at most 4/3. 

To summarize, the total cost of \RHM for serving a sequence $\sigma$ is no more than the total weight of items in $\sigma$, denoted by $W(\sigma)$. At the same time, the total weight of items in a bin by \opt is at most 1.63 which implies that the cost of \opt for serving $\sigma$ is at least $W(\sigma)/1.59$. We conclude that the competitive ratio of \RHM is at most 1.59 in this case.

\ \\ 

\textbf{Case 2:} Assume there is no $a_1$-bin in the packing. 
For the cost of \RRM for serving $\sigma_{cl_1}$, we have:
\begin{align}
\sigma_{cl_1}& = N_{d} + N_{c} + N_{ac} + N_{bc} + N_{bb} + N_{aa} + N_{a_2}  \nonumber \\ 
& \leq n_d + n_c + n_{b_2}/2 + N_{red} + 1 \nonumber
\label{eq:168}
\end{align} 

Recall that the algorithm ensures that $N_{red} \leq 3 N_{blue} + 3$. By Equation \ref{eq:6star}, we get 
$N_{red} \leq 3n_a/7+3n_{b_1}/7+3/7$. Plugging this into the above inequality, we will get:

\begin{align*}
\RRM(\sigma_{cl_1}) & \leq n_d + n_c + n_{b_2}/2 + 3n_a/7 + 3n_{b_1}/7 + 2 \nonumber \\
&< n_d + n_c + n_b/2 + 3n_a/7  + 2.
\end{align*} 

By Lemma \ref{lemTotCost}, for the total cost of \RHM, we will have:

\begin{center}
\begin{equation}
\RHM(\sigma) \leq \sum\limits_{k=2}^{18} \frac{n_k}{k+1} + 20W'/19 + n_d + n_c + n_b/2 + 3n_a/7 + 23 
\end{equation}
\label{eq:117}
\end{center}

%Similar to Case 1, we use a simple weighting function which ensures that the total weight of items is an upper bound for the cost of \RHM while the weight of any set of items which fit in a is at most 373/228. Consequently, the ratio between the cost of \RHM and that of \OPT is at most 228/373. The details of the weighting function can be found in the appendix.

Similar to Case 1, we use a weighting technique. For all items, except for $a$, $b$, and $c$ items, the weights are defined similar to Case 1. For $a$, $b$, and $c$ items, the weights are respectively $3/7$, $1/2$, and $1$. As Inequality \ref{eq:117} suggests, this definition for weights ensures that the total weight of items is an upper bound for the cost of \RHM. As before, we study the maximum weight of a bin in the packing of \opt; let $\beta_1\geq \beta_2 \geq \ldots \geq \beta_t$ be the items in such a bin and $W_{opt}$ be their total weight. We claim that $W_{opt} < 1.63$. Note that the density of $d$ items and larger items are at most $96/59<1.63$, the density of $c$, $b$ and $a$ items are respectively upper bounded by 2, 1.3, and 1.29, and the density of items smaller than $1/3$ which belongs to class $k\geq 2$ is at most $\frac{k+2}{k+1}$ except the last class (class 19) for which the density is at most $20/19$.

We do a case analysis as before. First, assume $\beta_1$ is not a $c$-item. In this case, the density of all items, and consecutively their total weight, is less than 1.63. Next, assume $\beta_1$ is a $c$ item.
Now, if $\beta_2$ is a $b$ item, the total weight of other items will be at most $11/96<1/8$. Hence, these items belong to class 7 or higher and their density is at most 9/8. The total weight of items in the bin will be $1+1/2 +11/96 \times 9/8 \leq 1.62$. Next, assume $\beta_2$ is an $a$ item; the total weight of other items will be at most 1/6. These items belong to class 5 or higher and their density is at most 7/6. Hence, the total weight of items will be at most $1+ 3/7 + 1/6 \times 7/6 <1.63$. Next, assume $\beta_2$ belongs to class 2; the total weight of other items will be at most 1/4. Now, if $\beta_3$ belongs to class 3, the weight of other items will be at most 1-1/2-1/4-1/5=1/20; so they belong to the last class and their density is at most 20/19. The total weight of items will be at most $1 + 1/3 + 1/4 + 1/20 \times 20/19 = 373/228 \approx 1.636$. If $\beta_3$ belongs to class 4 or higher, its density will be at most 6/5, and the total weight of items in the bin will be at most $1+1/3+1/4\times 6/5 < 1.634$. Finally, assume $\beta_2$ belongs to class 3 or higher; so, all items except $\beta_1$ have a density of at most 5/4. The total weight of items will be at most $1+ 1/2 \times 5/4 = 1.625$. 

To summarize, the total weight of items in a bin in \opt's packing is at most 373/228 which implies that the cost of \opt for serving $\sigma$ is at least $228 W(\sigma)/373$. Recall that the cost of \RHM for serving $\sigma$ is upper bounded by $W(\sigma)$. We conclude that the competitive ratio of \RHM is at most 373/228 in this case.

\end{proof}

\subsection{Average-Case Analysis}\label{khar}
We show that the average-case performance of \RHM is as good as \BF, \FF, and \HM. Except the following lemma, other aspects of the proof are similar to those in Section \ref{avgHmSection}.

\begin{lemma}\label{lemRRM}
For any instance $\sigma$ of the bin packing problem in which items are in range $(1/3,2/3]$, the cost of \RRM for serving $\sigma$ is no more than that of Matching Best Fit (\MBF).
\end{lemma}

\begin{proof}

The key observation is that \RRM uses \BF strategy to place a small item $x$ in a bin which includes a large item; and only if such a bin does not exist, it deviates from the \BF strategy. Let $S^t_{RRM}$ and $S^t_{MBF}$ respectively denote the set of large items which are not accompanied by a small item in the packings maintained by \RRM and \MBF (respectively) after placing the first $t$ items in $\sigma$ ($0 \leq t \leq n$, where $n$ is the length of $\sigma$); we refer to these sets as \textit{single-set}s of the algorithms. We claim that for all values of $t$, a single-set of \RRM is a subset of that of \MBF, i.e., $S^t_{RRM} \subset S^t_{MBF}$. We prove this by induction. Note that for $t=0$ both single-sets are empty and the base case holds. Assume $S^t_{RRM} \subset S^t_{MBF}$ for some $t>0$ and let $x$ denote the $(t+1)$th item in $\sigma$. If $x$ is a large item, \MBF opens a bin for $x$, and $x$ will be included in the single-set for \MBF; $x$ may or may not be added to the single-set of \RRM (it will not be added if it is a $c$ item and there are $a_1$ bins in the packing). Regardless, we will have $S^{t+1}_{RRM} \subset S^{t+1}_{MBF}$. Next, assume $x$ is a small item. \MBF and \RRM both use \BF strategy to place $x$ in one of the bins in $S^t_{MBF}$ and $S^t_{RRM}$. If such a bin does not exist for \MBF, by induction hypothesis, it will not exist for \RRM and %. So, $S^{t+1}_{MBF} = S^{t}_{MBF}$ and $S^{t+1}_{RRM} = S^{t}_{RRm}$; consecutively 
the induction statement holds. Next, assume \MBF places $x$ in a bin $B \in S^{t}_{MBF}$; so, $B$ will be removed from single-set of \MBF, i.e., we have $S^{t+1}_{MBF} = S^{t}_{MBF} - \{B \}$. If $B \notin S^{t}_{RRM}$, the induction statement holds because $S^{t+1}_{RRM}$ will be a subset of $S^{t}_{RRM}$ which is a indeed a subset of $S^{t+1}_{MBF}$ (a non-common items is removed from $S^t_{MBF}$). IF $B \in S^{t}_{RRM}$, \RRM places $x$ in $B$; this is because, similar to \MBF, \RRM uses \BF strategy to place small items in bins which include large items. Consecutively, we have $S^{n}_{RRM} \subset S^{n}_{MBF}$. 

The cost of \MBF for serving $\sigma$ is $n_{small} + \left|S^{n}_{MBF}\right|$ in which $n_{small}$ is the number of small items in $\sigma$. This is because \MBF opens a new bin for each small item. In the final packing of \RRM, the number of bins which include small items is no more than $n_{small}$. Other bins in the packing are associated with items in $S^n_{RRM}$; since the single-set of \RRM is a subset of that of single-set of \MBF, we have $\left|S^n_{RRM}\right| \leq \left|S^n_{MBF}\right|$. Hence, in total, the number of bins in the packing of \RRM is no more than that of \MBF.
\end{proof}

\begin{theorem} \label{majorAvgRRM}
Let $\sigma$ be a sequence of length $n$ in which item sizes are selected uniformly at random from $(0,1]$. The expected wasted space of \RHM for packing $\sigma$ is $\Theta(\sqrt{n} \log ^{3/4} n)$.
\end{theorem}

\begin{proof}

Let $\sigma^-$ be a copy of $\sigma$ in which those items which are placed in mature bins are removed. % (these are small items that are placed by \HM in bins which have been mature before placing them).
Also, let $\sigma_2^-, \ldots, \sigma_{19}^-$ be the subsequences of $\sigma^-$ formed by items belonging to different classes of \HM. We have 

%\begin{equation*}
%\HM(\sigma) = \sum\limits_{t=1}^{K} \ROM(\sigma_t^-) \leq \sum\limits_{t=1}^{K} \MBF(\sigma_t^-) \leq \sum\limits_{t=1}^{K} \MBF(\sigma_t)
%\label{eq:1}
%\end{equation*}

%Let $\sigma_1, \ldots, \sigma_{19}$ be the subsequences formed by items belonging to different classes of \RHM. We have 

\begin{equation*}
\HM(\sigma) = \RRM(\sigma_1) + \sum\limits_{t=2}^{19} \ROM(\sigma_t^-) \leq \sum\limits_{t=1}^{19} \MBF(\sigma_t^-) \leq \sum\limits_{t=1}^{19} \MBF(\sigma_t)
\end{equation*}

The second-to-last inequality comes from Lemmas \ref{OMMBF} and \ref{lemRRM} and the last inequality comes from Lemma \ref{MBFMONO}. Consequently, by Lemma \ref{MBFEx}, we have:
\begin{equation*}
E[\HM(\sigma)] \leq \sum\limits_{t=1}^{19} \left(n_t/2 + \Theta(\sqrt{n_t} \log^{3/4}n_t) \right) = \frac{n}{2} + \Theta(\sqrt{n}\log^{3/4}n)
\end{equation*}

The expected cost of \OPT is $n/2$ (it is not better since half items are expected to be larger than 1/2). Consequently, $E[\HM(\sigma) - \OPT(\sigma)] =\Theta(\sqrt{n}\log^{3/4}n)$
which completes the proof. 
\end{proof}

%s before, since the expected cost of \OM is $n/2 + o(n)$, the average performance ratio of $\HM_K$ is 1.

\section{Experimental Evaluation}\label{sect:exp}
The results of the previous sections indicate that \HM and \RHM have similar average-case performance as \BF if we assume a uniform, \emph{continuous} distribution for item sizes. In this section, we further observe the performance of these algorithms on sequences which follow other distributions. In doing so, we experimentally compare \HM and \RHM against classical bin packing algorithms. Table \ref{bench} gives details of the datasets that we generated for our experiments. %generated to compare online bin packing algorithms. %In what follows, we briefly describe them. 
In all cases, $E$ indicates the size of the bins. For each set-instance, the item sizes are randomly taken from a subset of the set $\{1, 2, \ldots, E\}$ of integers; this subset defines the \textit{range} of the items in the set-instance. Typically, we have $E=1000$, and the range is $[1,1000]$. Here, we briefly describe the considered distributions:

\begin{itemize}
	\item Discrete Uniform Distribution (DU sequences): We test the algorithms on \emph{discrete} uniform distributions. It is known that average-case behavior of bin packing algorithms can be substantially different under discrete and continuous distributions \cite{CCGJMS91}. The bin packing problem is extensively studied under discrete uniform distributions (see, e.g., \cite{CCGJMS91,CsiWoe97,Albers98p2}).
	
	\item NORMAL and POISSON sequences: Normal and Poisson distributions are two natural alternatives for generating item sizes. Both of these distributions are previously studied for generating bin packing sequences (see, e.g., \cite{FloKar91,Still08}). %Using Poisson process to create bin packing instances has been in studied in \cite{FloKar91}.
		
	\item Zipfian Distribution (ZD sequences): In bin packing sequences that follow the Zipfian distribution, item sizes follow the power-law, i.e., a large number of items are pretty small while a small number of items are quite large. The distribution has a parameter $\theta$ ($0 < \theta < 1$) that indicates how skewed the distribution is. Bin packing sequences with Zipfian distribution are considered in the experiments in \cite{apple08}.
	
	\item SORTD sequences: To create an instance of this family, we take a sequence of uniformly random-sized items and sort it in decreasing order of item sizes. This way, we can compare the performance of \emph{offline} versions of the algorithms (where sequences are sorted in decreasing order before being packed in an online manner).	
	
		\item Weibull Distribution (WD sequences): In \cite{CastIg12}, it is shown that Weibull distribution can be used to model real-world bin packing benchmarks. The values considered for the shape parameter $k$ are among the ones suggested in \cite{CastIg12}. The scale parameter of the distribution is set to be proportional to $E$. %This specific setting of parameters is among the ones studied in \cite{CastIg12}.

	\item Bounded Probability Sampled Distributions (BPS sequences): 
	To get these sequences, a random distribution is generated as follows. Given a parameter $s$, we select $s$ random numbers in a given range and assign random weights to them. The probability associated with an item in the distribution is proportional to its weight. These sequences were first introduced in \cite{bpbench3} and later used in the experiments in \cite{apple08}.

\end{itemize}

\begin{small}
\begin{table*}[!t]
\begin{center}
  \begin{tabular}{|c|c|c|c|| }
  \hline
	Set-instance  													& Distribution										& $E$ & Range \\
%%***********************************************************************************************************************************************************	 
	\hline
	\hline
	 DU0						 												&	Uniform 														& 	100		& (1,E)		\\ %fix 	
	 DU1						 												&	Uniform 														& 	500		& (1,E)		\\ %fix 	
	 DU2						 												&	Uniform 														&1,000	& (1,E)		\\ %fix
 	 DU3						 												&	Uniform 														&1,000		& (1,E/2)		\\ %fix 
	 DU4						 												&	Uniform 														&1,000		& (1,E/10)	\\ %fix
	 NORMAL																	& Normal ($\mu = E/2$, $\sigma = E/6$)	&1,000		&		(1,E)		\\%fix
	 POISSON																& Poisson ($\lambda = E/3$)							& 	1,000		&		(1,E)		\\%fix
	 ZIPF1																	& Zipfian ($\theta = 1/2$)						& 	1,000		&		(1,E)		\\%fix
	 ZIPF2																	& Zipfian ($\theta = 1/3$)							& 	1,000		&		(1,E)		\\%fix
	SORTD																		&	Uniform, sorted decreasing				&	1,000		& (1,E)		\\ %fix 	
	WD1																		& Weibull ($k = 0.454$)								&1,000		& (1,E)		 \\ %fix 	 
	WD2																		& Weibull ($k = 1.044$)									&		1,000		& (1,E)		 \\	%fix
	 BPSD1																		& BPS distribution ($s=100$)				& 1,000		& (E/4,E/2)  \\ % fix
	 BPSD2																		& BPS distribution ($s=100$)			& 1,000		& (1,E/4)  \\ % fix
 	 \hline 
	\end{tabular}
  \caption{\label{bench}. The distributions used to create set-instances to compare algorithms.}
\end{center}
\end{table*}
\end{small}

For each of the indicated set-instances, we create $1000$ random sequences of length $10^6$ and compute the average costs of different algorithms for packing these sequences. 
%Considering the length of the sequences, the indicated number of sequences is sufficient to achieve a convergence.
Beside Any-Fit and Harmonic family of algorithms, we also consider \textit{Sum-of-Squares (\algo{SS}) Algorithm} which performs well for discrete distributions \cite{Csi99,Csir05,Csir06}.
To place an item into a partial packing $P$, \algo{SS} defines \textit{sum of squares} of $P$, denoted by $ss(P)$, as $\sum_h N_P(h)^2$; here $N_P(h)$ denotes the number of bins with level $h$ in $P$. To place an item $x$, $SS$ places $x$ into an existing bin, or opens a new bin for $x$, so as to yield the minimum possible value of $ss(P')$ for the resulting packing $P'$. Note that \algo{SS} is not well-defined for the classical, continuous version of the bin packing problem. In \cite{Csir06}, it is proved that for any discrete distribution in which the optimal expected waste is sub-linear, \algo{SS} also has sub-linear expected waste. In particular, for those distributions where the optimal expected waste is constant (the so-called a \textit{perfect} distributions), \algo{SS} has an expected waste of at most $O(\log n)$. In our experiments, we treat \algo{SS} as the closest online algorithm to \opt. We note that the competitive ratio of \algo{SS} is at least 2 and at most 2.77 \cite{Csir05} which is worse than most online algorithms.

Figure \ref{fexp} shows the average costs (the number of opened bins) of the classical bin packing algorithms, as well as \HM and \RHM, for serving the above set-instances. For algorithms that classify items by their sizes (e.g., \HA and \HM), the number of classes (the value of $K$) is set to 20. The results in Figure \ref{fexp} indicate that in all cases \HM and \RHM perform significantly better than other members of Harmonic family. At the same time, they have comparable performance with \BF and \FF. In most cases, \HM and \RHM perform better than \FF, and in some cases, e.g., NORMAL set, they even perform better than \BF. 

Comparing the costs of \HM and \RHM for DU0, DU1, and DU2, we observe that their relative performance improves when the size of the bins (i.e., $E$) increases. For small value of $E=100$ (UD0), these algorithms are slightly worse than \FF. However, as $E$ increases to $1000$ (UD2), the algorithms perform better than \FF and converge to \BF. This is in accordance with the results in Sections \ref{avgHmSection} and \ref{khar} which imply that, for continues uniform distribution, the expected costs of \HM and \RHM converge to that of \BF. Note that as $E$ goes to infinity, the discrete distribution estimates a continuous one.

For symmetric distributions, where items of sizes $x$ and $E-x$ appear with the same probability, the expected costs of \HM and \RHM are equal. A difference between the packings of \HM and \RHM happens when a number of small items of the first class (items of type $a$ in \RHM) appear before any large item of the same class (an item of type $c$). In these cases, \RHM `reserves' some bins for subsequent large items (by declaring the bins as being blue). For symmetric distributions, however, it is unlikely that many small items appear before the next large item. Consequently, as the numbers in Figure \ref{fexp} reflect, the average costs of \HM and \RHM are the same for symmetric sequences. On the other hand, for asymmetric sequences where small items are more likely to appear, e.g., DU3 and POISSON, \HM has an advantage over \RHM. In these sequences, there is no reason to reserve bins for the large items since they are unlikely to appear. %On the other hand, \RHM had advantage over \HM and other algorithms for SORI (where items are uniformly random, butt sorted in increasing order). For these sequences, reserving space for large items indeed helps as they appear at the end of the sequence.

Finally, we note that for SORTD, the costs of \HM and \RHM are comparable to those of \BF and \FF. This implies that, on average, the offline versions of these algorithms are comparable with the well-known First-Fit-Decreasing and Best-Fit-Decreasing algorithms. It remains open whether the same statement holds for the worst-case performance of the offline algorithms. % (in terms of approximation ratio).

\begin{figure}[!t]
\centering
\includegraphics[width=\columnwidth,trim = 2.5cm 33cm 3.5cm 0mm, clip]{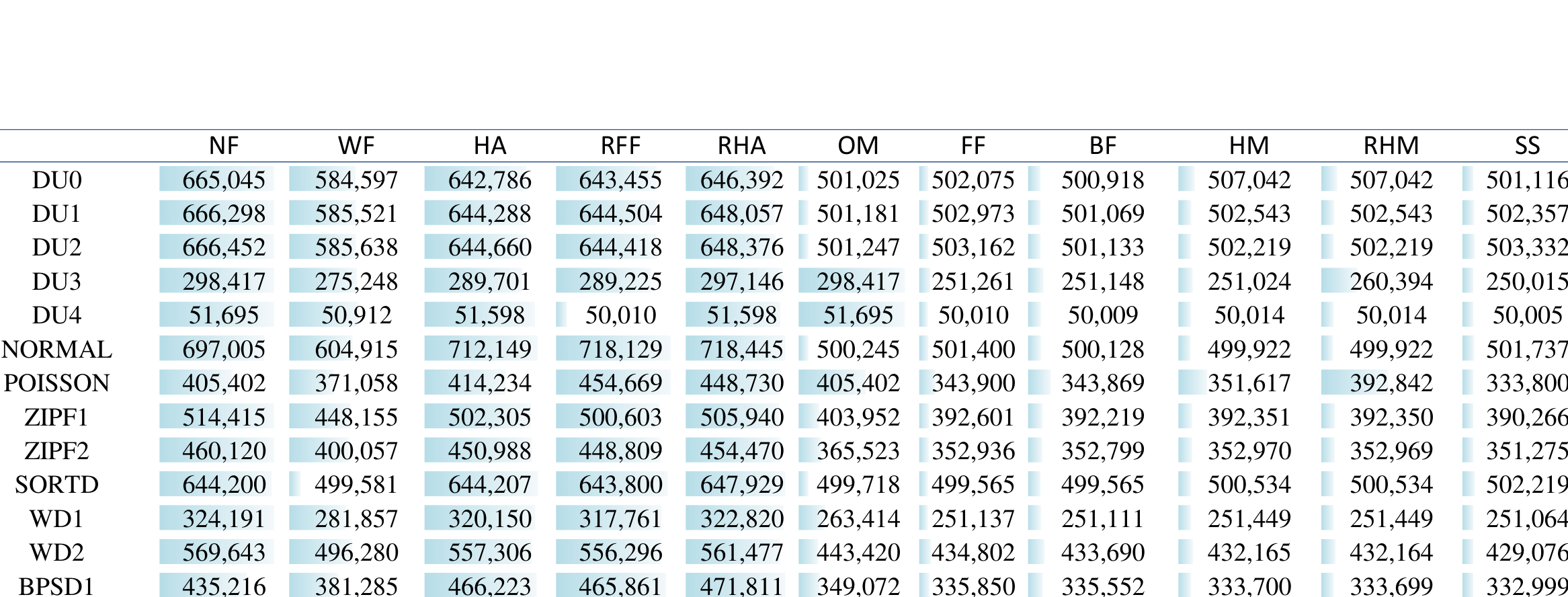}\label{fig:results}
\caption{Average performance of online bin packing algorithms for different set-instances. 
The indicated numbers for each algorithm represent the average costs of the algorithm for different set-instances. In most cases, there is a gap between the cost of \HM (and \RHM) and other Harmonic-based algorithms. 
The data-bar of an algorithm for a set-instance indicates the \emph{rank} of the cost of the algorithm, among the costs of all algorithms, for serving that set-instance. }%; the shorter the barcode is, the better the rank of the algorithm is. }
\label{fexp}%
\end{figure}

\section{Discussion}
%We introduced the \algo{Harmonic Match} algorithm for the online bin packing problem. The competitive ratio of this algorithm approaches $T_\infty \approx 1.691$ which is slightly better than 1.7 of \algo{First Fit} and \algo{Best Fit}. Meanwhile, the average performance of the algorithm is better than \FF and as good as \BF. We also introduced Refined Harmonic Match algorithm which improves the competitive ratio to 1.636, while maintaining the good average-case performance of \HM. In fact, 

\HM and \RHM can be seen as variants of Harmonic and Refined Harmonic algorithms in which small and large items are carefully matched in order to improve the average performance, while preserving the worst-case performance. We believe that the same approach can be applied to improve the average performance of other Super Harmonic algorithms, and in particular that of Harmonic++ algorithm (which is currently the best online bin packing algorithm, regarding the competitive ratio). Given the complicated nature of these algorithm, modifying them involves a detailed analysis which we leave as a future work.

The relative worst order analysis is an alternative method for comparing online algorithms. This method considers the worst ordering of a given sequence for two online algorithms and indicates
 their costs on these orderings. Then, among all sequences, it considers the one that maximizes the worst-case ratio between the two algorithms (see \cite{BoyaFav07} for a precise definition). It is known that under relative worst order analysis, First Fit is no worse than any Any Fit algorithm (and in particular Best Fit) \cite{BoyaFav07}. Also, Harmonic algorithm is not comparable to \FF for sequences which include very small items. However, when all items are larger than $\frac{1}{K+1}$ ($K$ is the parameter of the Harmonic algorithm), Harmonic is better than \FF by a factor of 6/5 \cite{BoyaFav07}. Applying Theorem \ref{th1}, we conclude that when all items are larger than $\frac{1}{K+2}$, Harmonic Match with parameter $K$ is strictly better than \FF and \BF under the relative worst order analysis. This provides another evidence for the advantage of Harmonic Match over \BF and \FF. 

Many online bin packing algorithms, e.g., \BF and \FF, have the undesired property that removing an item might increase the cost of the algorithm \cite{Murgolo88}. This `anomalous' behavior results in an unstable algorithm which is harder to analyze. %We show that \HM is a \textit{monotone} algorithm, i.e.,
As mentioned earlier, an algorithm is called `monotone' if removing an item does not increase its cost. 
%This property might have practical significance while at the same time facilitates the theoretical analysis of the algorithm. % (e.g., we used monotony of various algorithms in our analysis in this paper). 
It is not clear whether \HM and \RHM are monotone; however, a slight twist in \HM results in a monotone algorithm. Consider a modified algorithm \HMM that works similar to \HM except that it does not maintain mature bins, i.e., it closes a bin as soon as it becomes mature. It is not hard to see that \HMM is a monotone algorithm. At the same time, the results related to the worst-case and average-case performance of \HM hold also for \HMM (Corollary \ref{col1} and Theorem \ref{majorAvg}). However, \HMM performs slightly worse than \HM on discrete distributions evaluated in Section \ref{sect:exp}. We leave further analysis of monotonous behaviour of this algorithm as a future work.

%Although we only discussed the uniform distribution, the algorithm \algo{Harmonic Match} is expected to perform well for other \textit{symmetric} distributions. A distribution is symmetric if the chance of an item having size $x$ is equal to having size $1-x$. In \cite{RheTal93B}, an algorithm is suggested which is closely related to \BF and has expected waste of size $\Theta(\sqrt{n} \log ^{3/4}n)$ for any distribution. However, the algorithm performs poorly in the worst case; one might consider applying the techniques of this paper to achieve better worst-case behavior while maintaining the good average-case performance for symmetric distributions. %We should mention that our analysis does not work for asymmetric distributions, in particular when items are uniformly distributed in the range $(0,\alpha]$ where $\alpha <1$. However, one might modify the Harmonic Match algorithm (in particular its interval definitions) to adapt it with an asymmetric distribution. This requires fundamental adjustment to the algorithm and we leave it as a future project.
%
\newpage

\bibliographystyle{splncs03}
\bibliography{confs,online}

\newpage

\end{document}